\documentclass[english,aps,pra,amsfonts,amssymb,amsmath,twocolumn,groupedaddress,nofootinbib]{revtex4-1}
\usepackage[utf8]{inputenc}
\pdfpageattr{/Group <</S /Transparency /I true /CS /DeviceRGB>>}
\usepackage{amsthm}
\usepackage{amsmath}
\usepackage{amssymb}
\usepackage{subcaption}

\usepackage{tikz}
\usetikzlibrary{automata,positioning}

\PassOptionsToPackage{normalem}{ulem}
\usepackage{ulem}

\usepackage{comment}
\usepackage{caption}
\usepackage{algorithm}
\usepackage{algorithmic}
\usepackage{setspace}

\usepackage{float}
\restylefloat{table}

\usepackage[unicode=true,pdfusetitle,
 bookmarks=true,bookmarksnumbered=false,bookmarksopen=false,
 breaklinks=false,pdfborder={0 0 1},backref=false,colorlinks=false]
 {hyperref}
\makeatletter

\theoremstyle{plain}
\newtheorem{thm}{\protect\theoremname}
\theoremstyle{plain}
\newtheorem{prop}[thm]{Proposition}
\newtheorem{lem}[thm]{Lemma}
\newtheorem{defin}[thm]{\protect\definname}
\newtheorem{observ}[thm]{\protect\observname}
\newtheorem{corol}[thm]{\protect\corolname}

\newtheorem{lemma}[thm]{Lemma}

\theoremstyle{definition}

\theoremstyle{remark}
\newtheorem{remark}[thm]{Remark}

\newcommand{\CC}{\mathbb C}

\renewcommand{\phi}{\varphi}

\newcommand{\ket}[1]{|#1\rangle}

\newcommand\nix[1]{{}}

\makeatother

\usepackage{babel}

\providecommand{\theoremname}{Theorem}
\providecommand{\definname}{Definition}
\providecommand{\observname}{Observation}
\providecommand{\corolname}{Corollary}
\providecommand{\examplename}{Example}
\providecommand{\problemname}{Problem}
\usepackage{pgfplots}
\usetikzlibrary{calc}
\usepackage{pgfplotstable}
\usepackage{colortbl}
\usepackage{booktabs}
\usepackage{pgfplotstable}

\begin{document}

\title{Factoring with Qutrits: Shor's Algorithm on Ternary and Metaplectic Quantum Architectures}

\author{Alex Bocharov$^*$}
\author{Martin Roetteler$^*$}
\author{Krysta M.~Svore$^*$}

\affiliation{
$^*$Quantum Architectures and Computation Group, Station Q, Microsoft Research, Redmond, WA (USA) \\}

\begin{abstract}

We determine the cost of performing Shor's algorithm for integer factorization on a ternary quantum computer, using two natural models of universal fault-tolerant computing:
 %on ternary quantum systems:
 (i) a model based on magic state distillation that assumes the availability of the ternary Clifford gates, projective measurements, classical control as its natural instrumentation set; (ii) a model based
 on a metaplectic topological quantum computer (MTQC). A natural choice to implement Shor's algorithm on a ternary quantum computer is to translate the entire arithmetic into a ternary form. However, it is
 also possible to emulate the standard binary version of the algorithm by encoding each qubit in a three-level system. We compare the two approaches and analyze the complexity of implementing Shor's period
 finding function in the two models.
We also highlight the fact that the cost of achieving universality through magic states in MTQC architecture is asymptotically lower than in generic ternary case.
\end{abstract}

\maketitle

\section{Introduction}

Shor's quantum algorithm for integer factorization \cite{Shor} is a striking case of superpolynomial speed-up promised by a quantum computer over the best-known classical algorithms.
%It had arguably the greatest influence on academic research and engineering efforts in the field of quantum computation and quantum information.
Since Shor's original paper, many explicit circuit constructions over qubits for performing the algorithm have been developed and analyzed.
This includes automated synthesis of the underlying quantum circuits for the binary case (see the following and references therein:
\cite{BeckmanEtAl,Beauregard,CleveWatrous,Markov2012,Markov2013,MeterItoh,TakaKuni,VedralEtAl,Zalka}).

It has been previously noted that arithmetic encoding systems beyond binary may yield more natural embeddings for some computations and potentially lead to more efficient solutions.
(A brief history note on this line of thought can be found in section 4.1 of \cite{Knuth}.)
Experimental implementation of computation with ternary logic, for example with Josephson junctions, dates back to 1989 \cite{morisue1989jctl,morisue1998memory}. More recently, multi-valued logic has been
proposed for linear ion traps \cite{Muthu2000mvl}, cold atoms \cite{smith2013cs}, and entangled photons \cite{malik2016multiPhoton}. %It remains to be seen, at what scale it would be possible to balance out
%quantum universality and fault-tolerance in these and other similar architectures.
In topological quantum computing it has been shown that metaplectic non-Abelian anyons \cite{CuiWang} naturally align with ternary, and not binary, logic.
These anyons offer a natively topologically protected universal set of quantum gates (see, for example, \cite{NayakFreedman}), in turn requiring little to no quantum error correction.

It is also interesting to note that qutrit-based computers are in certain sense space-optimal among all the qudit-based computers with varying local quantum dimension. Thus in
\cite{GreentreeEtAl} an argument is made that, as the dimension of the constituent qudits increases, the cost of maintaining a qudit in fully entangled state also increases and the optimum cost
per Hilbert dimension is attained at local dimension of $\lceil e \rceil = 3$.

Transferring the wealth of multi-qubit circuits to multi-qutrit framework is not straightforward. Some of the binary primitives, for example the binary Hadamard gate and the two-qubit $\mbox{CNOT}$ gate, do
not remain Clifford operations in the ternary case.  Therefore, they cannot be emulated by ternary Clifford circuits. We resolve this complication by developing efficient non-Clifford circuits for a
\emph{generic} ternary quantum computer first.  We then extend the solution to the Metaplectic Topological Quantum Computer (MTQC) platform \cite{CuiWang}, which further reduces the cost of implementation.

A generic ternary framework that supports the full ternary Clifford group, measurement, and classical control  \cite{CampbellEtAl}, also supports a distillation protocol that prepares magic states for the
$P_9$ gate:
\begin{equation} \label{eq:def:P:9}
P_9 = \omega_9^{-1}\,|0\rangle \langle 0|+|1\rangle \langle 1|+\omega_9\,|2\rangle \langle 2|, \, \omega_9 = e^{2 \pi \, i/9}.
\end{equation}
The Clifford+$P_9$ basis is universal for quantum computation and serves a similar functional role in ternary logic as the Clifford+$T$ basis in binary logic (see \cite{CampbellEtAl,howard2012qudit} for the
more general qudit context).

We show in more detail further that the primitive $R$ gate available in MTQC is more powerful in practice than the $P_9$ gate.

Arguably, a natural choice to implement Shor's algorithm on a ternary quantum computer is to translate the entire arithmetic into ternary form. We do so by using ternary arithmetic tools developed in
\cite{BCRS} (with some practical improvements). We also explore alternative approach: emulation of binary version of Shor's period finding algorithm on ternary processor. Emulation has notable practical advantages in some contexts. For example, as shown in section \ref{subsec:ripple:carry}, using a binary ripple-carry additive shift consumes fewer clean $P_9$
magic states than the corresponding  ternary ripple-carry additive shift (cf. table \ref{table:binary:vs:ternary}).

We also show that on a metaplectic ternary computer the magic state coprocessor is asymptotically smaller than a magic state distillation coprocessor, such as the one developed in
\cite{CampbellEtAl} for the generic ternary quantum computer. Another benefit of the MTQC is the ability to approximate desired non-Clifford reflections directly to the required fidelity, thus eliminating the need for magic states. The tradeoff is an increase in the depth of the emulated Shor's period finding circuit by a logarithmic factor, which is tolerable for the majority of instances.

The cost benefits of using exotic non-Abelian anyons for integer factorization has been previously noted, for example in \cite{BarabanEtAl}, where hypothetical
Fibonacci anyons were used. It is worthwhile noting that neither binary nor ternary logic is native to Fibonacci anyons, so the $\mbox{NOT}$, $\mbox{CNOT}$ or Toffoli gates are
much harder to emulate there than on a hypothetical metaplectic anyon computer.

The paper is organized as follows.
In Section \ref{sec:backgroud} we state the definitions and equations pertaining the two ternary architectures used, and gaive a quick overview of the Shor's period finding function.
In Section \ref{sec:Emulation} we perform a detailed analysis of reversible classical circuits for modular exponentiation. We compare two designs of the modular
exponentiation arithmetic. One is emulation of binary encoding of integers combined with ternary arithmetic gates. The other uses ternary encoding of integers with ternary gates.

In Section \ref{sec:reflections} we develop circuits for the  key arithmetic gates based on designs from \cite{BCRS} with further optimizations.

In Section \ref{sec:specific:resource:counts} we compare the resource cost of performing modular exponentiation. An interesting feature of
ternary arithmetic circuits is the fact that the denser and more compact ternary encoding of integers does not necessarily lead to more resource-efficient period finding solutions compared to binary encoding.
The latter appears to be better suited in practice for low-width arithmetic circuit designs (hence, e.g., for smaller quantum computers).

We also compare the magic state preparation requirements. We highlight the huge advantage of the metaplectic topological computer. Magic state
preparation requires width that is linear in $\log(n)$ on an MTQC, whereas it requires width in $O(\log^3(n))$ on a generic ternary quantum computer. \footnote{It requires width in
%change MR
$O(\log^{\gamma}(n))$ in the binary Clifford+T architecture, where $\gamma$ can vary between $\log_2(3)$ and $\log_3(15)$ depending on practically applicable distillation protocol.}

All the circuit designs and resource counts are done under assumption of fully-connected multi-qutrit network. Factorization circuitry optimized for sparsely connected networks, such as nearest-neighbor for example, is undeniably interesting (cf. \cite{PhamSvore}) but we had to set this topic aside in the scope of this paper.

\section{Background and Notation} \label{sec:backgroud}

%``Dixitque Deus fiat reflexio.'' \cite{Synthesis}

A common assumption for a multi-qudit quantum computer architecture is the availability of quantum gates generating the full multi-qudit Clifford group (see \cite{howard2012qudit},\cite{CampbellEtAl}). In this
section we describe a \emph{generic} ternary computer, where the full ternary Clifford group is postulated; we also describe the more specific Metaplectic Topological Quantum Computer (MTQC) where the required
Clifford gates are explicitly implemented by braiding non-Abelian anyons \cite{CuiHongWang, CuiWang}. For purposes of this paper, each braid corresponds to a unitary operation on qutrits. Braids are considered relatively inexpensive and tolerant to local noise. Universal quantum computation on MTQC is achieved by adding a single-qutrit phase flip gate
($\mbox{Flip}$ in \cite{CuiWang}, $R_{|2\rangle}$ in \cite{BCKW} and our Subsection \ref{subsec:metaplectic:basis}). In contrast with the binary phase flip $Z$, which is a Pauli gate, the ternary phase flip is
not only non-Clifford, but it does not belong to \emph{any} level of Clifford hierarchy (see, for example, \cite{BCRS}). Intuitively one should expect this gate to be very powerful. Level $\mathcal{C}_3$ of the ternary Clifford hierarchy is emulated quite efficiently on MTQC architecture, while the converse is quite expensive: implementing phase flip in terms of $\mathcal{C}_3$ requires
several ancillas and a number of repeat-until-success circuits.

\subsection{Ternary Clifford group}
Let $\{|0\rangle,|1\rangle,|2\rangle\}$ be the standard computational basis for a qutrit.
Let $\omega_3 = e^{2 \pi \, i/3}$ be the third primitive root of unity.
The ternary \emph{Pauli} group is generated by the \emph{increment} gate
\begin{equation} \label{eq:INC}
\mbox{INC}= |1\rangle \langle 0|+|2\rangle \langle 1|+|0\rangle \langle 2|,
\end{equation}
\noindent and the ternary $Z$ gate
\begin{equation} \label{eq:Z}
Z= |0\rangle \langle 0|+\omega_3 |1\rangle \langle 1|+\omega_3^2 |2\rangle \langle 2|.
\end{equation}

The ternary \emph{Clifford} group stabilizes the Pauli group is obtained by adding the ternary Hadamard gate $H$,
\begin{equation} \label{eq:H}
H = \frac{1}{\sqrt{3}} \sum \omega_3^{j\, k} |j\rangle \langle k|,
\end{equation}
\noindent the $Q$ gate
\begin{equation} \label{eq:Q}
Q= |0\rangle \langle 0|+|1\rangle \langle 1|+\omega_3 |2\rangle \langle 2|,
\end{equation}
\noindent and the two-qutrit $\mbox{SUM}$ gate,
\begin{equation} \label{eq:SUM}
\mbox{SUM} |j,k\rangle = |j,j+k \; {\rm mod} \; 3\rangle, j,k \in \{0,1,2\}
\end{equation}
\noindent to the set of generators of the Pauli group.

Compared to the binary Clifford group, $H$ is the ternary counterpart of the binary Hadamard gate, $Q$ is the counterpart of the phase gate $S$, and $\mbox{SUM}$ is an analog of the $\mbox{CNOT}$ (although,
intuitively it is a ``weaker'' entangler than $\mbox{CNOT}$, as described below).

For any $n$, ternary Clifford gates and their various tensor products generate a finite subgroup of $U(3^n)$; therefore they are not sufficient for universal quantum computation.
We consider and compare two methods of building up quantum universality: by implementing the $P_9$ gate as per Eq. (\ref{eq:def:P:9}) and by expanding into the metaplectic basis (Subsection
\ref{subsec:metaplectic:basis}).
Given enough ancillae, these two bases are effectively and efficiently equivalent in principle (see Appendix \ref{app:sec:R:2}), and the costs in ancillae create practical tradeoffs depending on the given
application.

\subsection{Binary and ternary control}

Given an $n$-qutrit unitary operator $U$ there are different ways of expanding it into an $(n+1)$-qutrit unitary using the additional qutrit as ``control''.
Let $|c\rangle$ be a state of the control qutrit and $|t\rangle$ be a state of the $n$-qutrit register. We define
\[
C_{\ell}(U)|c\rangle |t\rangle = |c\rangle \otimes (U^{\delta_{c,\ell}}) |t\rangle,\, \ell \in \{0,1,2\},
\]
\noindent wherein $\delta$ denotes the Kronecker delta symbol.
We refer to this operator as a \emph{binary-controlled} unitary $U$ and denote it in circuit diagrams as

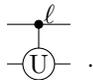
\begin{figure}[h!]
\centering
\begin{tikzpicture}[scale=1.000000,x=1pt,y=1pt]
\filldraw[color=white] (0.000000, -7.500000) rectangle (24.000000, 22.500000);
% Drawing wires
% Line 1: a W
\draw[color=black] (0.000000,15.000000) -- (24.000000,15.000000);
% Line 2: b W
\draw[color=black] (0.000000,0.000000) -- (24.000000,0.000000);
% Done with wires; drawing gates
% Line 3: b P U a
\draw (12.000000,15.000000) -- (12.000000,0.000000);
\begin{scope}
\draw[fill=white] (12.000000, 0.000000) circle(6.000000pt);
\clip (12.000000, 0.000000) circle(6.000000pt);
\draw (12.000000, 0.000000) node {U};
\end{scope}
\filldraw (12.000000, 15.000000) circle(1.500000pt);
\draw (16.000000, 19.000000) node {$\ell$};
% Done with gates; drawing ending labels
% Done with ending labels; drawing cut lines and comments
% Done with comments
\end{tikzpicture}
\;\raisebox{2.3mm}{.}
\end{figure}

We omit the label $\ell$ when $\ell=1$. We also define the \emph{ternary-controlled} extension of $U$ by
\[
\Lambda(U)|c\rangle |t\rangle = |c\rangle \otimes (U^c\, |t\rangle)
\]
\noindent and denote it in circuit diagrams as

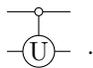
\begin{figure}[h!]
\centering
\begin{tikzpicture}[scale=1.000000,x=1pt,y=1pt]
\filldraw[color=white] (0.000000, -7.500000) rectangle (24.000000, 22.500000);
% Drawing wires
% Line 1: a W
\draw[color=black] (0.000000,15.000000) -- (24.000000,15.000000);
% Line 2: b W
\draw[color=black] (0.000000,0.000000) -- (24.000000,0.000000);
% Done with wires; drawing gates
% Line 3: b P U a
\draw (12.000000,15.000000) -- (12.000000,0.000000);
\begin{scope}
\draw[fill=white] (12.000000, 0.000000) circle(6.000000pt);
\clip (12.000000, 0.000000) circle(6.000000pt);
\draw (12.000000, 0.000000) node {U};
\end{scope}
\draw[fill=white] (12.000000, 15.000000) circle(1.500000pt);
%\draw (16.000000, 19.000000) node {$\ell$};
% Done with gates; drawing ending labels
% Done with ending labels; drawing cut lines and comments
% Done with comments
\end{tikzpicture}
\;\raisebox{2.3mm}{.}
\end{figure}

It is paramount to keep in mind that
\[
\mbox{SUM} = \Lambda(\mbox{INC})
\]
\noindent (see equations (\ref{eq:INC}) and (\ref{eq:SUM})).
Another useful observation is that for any unitary $U$ we have that $\Lambda(U) = C_1(U) \, (C_2(U))^2$.

More detail can be found in \cite{BCRS}.

\subsection{The $P_9$ gate and its corresponding magic state} \label{subsec:P9:gate}

It is easy to see that the $P_9$ gate in Eq. (\ref{eq:def:P:9}) is not a Clifford gate, e.g., it does not stabilize the ternary Pauli group. However, it can be realized by a certain deterministic measurement-assisted
circuit given a copy of the \emph{magic} state
\begin{equation} \label{eq:def:mu}
\mu = \omega_9^{-1} |0\rangle + |1\rangle+\omega_9 |2\rangle, \, \omega_9 = e^{2 \pi \, i/9}.
\end{equation}

An appropriate deterministic magic state injection circuit, as proposed in Ref.~\cite{CampbellEtAl}, is shown in Figure \ref{fig:mu:state:injection}.
\begin{figure}[h!]
\begin{tikzpicture}[scale=2.000000,x=1pt,y=1pt]
\filldraw[color=white] (0.000000, -7.500000) rectangle (72.000000, 22.500000);
% Drawing wires
% Line 2: a W |\psi\rangle
\draw[color=black] (0.000000,15.000000) -- (36.000000,15.000000);
\draw[color=black] (36.000000,14.500000) -- (60.000000,14.500000);
\draw[color=black] (36.000000,15.500000) -- (60.000000,15.500000);
\draw[color=black] (0.000000,15.000000) node[left] {$|\mu\rangle$};
% Line 3: b W |\mbox{input}\rangle
\draw[color=black] (0.000000,0.000000) -- (72.000000,0.000000);
\draw[color=black] (0.000000,0.000000) node[left] {$|\mbox{input}\rangle$};
% Done with wires; drawing gates
% Line 5: b P $\mbox{INC}$ a
\draw (12.000000,15.000000) -- (12.000000,0.000000);
\begin{scope}
%\draw[fill=white] (12.000000, 0.000000) circle(6.000000pt);
%\clip (12.000000, 0.000000) circle(6.000000pt);
%\draw (12.000000, 0.000000) node {$\mbox{INC}^{\dagger}$};
%\draw[fill=white] (12.000000, 15.000000) circle(1.500000pt);
\draw[fill=white] (12.000000, 15.000000) circle(6.000000pt);
\clip (12.000000, 15.000000) circle(6.000000pt);
\draw (12.000000, 15.000000) node {$\mbox{INC}^{\dagger}$};
\end{scope}
\draw[fill=white] (12.000000, 0.000000) circle(1.500000pt);
% Line 6: a M
\draw[fill=white] (30.000000, 9.000000) rectangle (42.000000, 21.000000);
\draw[very thin] (36.000000, 15.600000) arc (90:150:6.000000pt);
\draw[very thin] (36.000000, 15.600000) arc (90:30:6.000000pt);
\draw[->,>=stealth] (36.000000, 9.600000) -- +(80:10.392305pt);
% Line 7: b Z a:owire
\draw (59.500000,15.000000) -- (59.500000,0.000000);
\draw (60.500000,15.000000) -- (60.500000,0.000000);
\begin{scope}
\draw[fill=white] (60.000000, -0.000000) +(-45.000000:8.485281pt and 8.485281pt) -- +(45.000000:8.485281pt and 8.485281pt) -- +(135.000000:8.485281pt and 8.485281pt) -- +(225.000000:8.485281pt and 8.485281pt)
-- cycle;
\clip (60.000000, -0.000000) +(-45.000000:8.485281pt and 8.485281pt) -- +(45.000000:8.485281pt and 8.485281pt) -- +(135.000000:8.485281pt and 8.485281pt) -- +(225.000000:8.485281pt and 8.485281pt) -- cycle;
\draw (60.000000, -0.000000) node {$C_{\mu,m}$};
\end{scope}
\filldraw (60.000000, 15.000000) circle(1.500000pt);
% Done with gates; drawing ending labels
\draw[color=black] (72.000000,0.000000) node[right] {$P_9 |\mbox{input}\rangle$};
% Done with ending labels; drawing cut lines and comments
% Done with comments
\end{tikzpicture}
\caption{\label{fig:mu:state:injection} Exact representation of the $P_9$ gate by state injection. $C_{\mu,m}$ stands for a certain precompiled ternary Clifford gate, classically predicated by the measurement
result $m$. }
\end{figure}
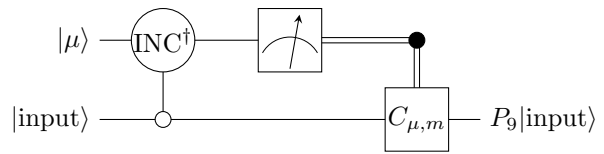
For completeness,  $C_{\mu,m}=(P_9 \, \mbox{INC} \, P_9^{\dagger})^{-m}\, \mbox{INC}^m$. Note that $P_9 \, \mbox{INC} \, P_9^{\dagger}$ is a Clifford gate, since $P_9$ is at level 3 of the ternary Clifford
hierarchy (cf.~\cite{BCRS}).

Such magic state naturally exists in any multi-qudit framework with qudits of prime dimension \cite{CampbellEtAl}. When the framework supports the full multi-qudit Clifford group, projective measurements and
classical control, then it also supports stabilizer protocols for magic state distillation based on generalized Reed-Muller codes. In particular, a  multi-qutrit framework supports a distillation protocol that
requires
$O(\log^3(1/\delta))$ raw magic states of low fixed fidelity in order to distill a copy of the magic state $\mu$ at fidelity $1-\delta$.
The distillation protocol is iterative and converges to that fidelity in $O(\log(\log(1/\delta)))$ iterations.
The protocol performance is analogous to the magic state distillation protocol for the $T$ gate in the Clifford+T framework \cite{BravyiKitaev}.

 One architectural design is to split the actual computation into ``online'' and ``offline'' components where the main part of quantum processor runs the target quantum circuit whereas the (potentially rather
 large) ``offline'' coprocessor distills copies of a magic state that are subsequently injected into the main circuit by a deterministic widget of constant depth.
Discussing the details of the distillation protocol for the magic state $\mu$ is beyond the scope of this paper and we refer the reader to Ref. \cite{CampbellEtAl}.

\subsection{Metaplectic quantum basis} \label{subsec:metaplectic:basis}

The ternary \emph{metaplectic} quantum basis is obtained by adding the \emph{single qutrit axial reflection} gate

\begin{equation} \label{eq:def:R:2}
R_{|2\rangle} = |0\rangle \langle 0| + |1\rangle \langle 1| - |2\rangle \langle 2|
\end{equation}

\noindent to the ternary Clifford group.
It is easy to see that $R_{|2\rangle}$ is a non-Clifford gate and that Clifford+$R_{|2\rangle}$ framework is universal for quantum computation.

In Ref.~\cite{CuiWang} this framework has been realized with certain weakly integral non-Abelian anyons called \emph{metaplectic anyons} which explains our use of the ``metaplectic'' epithet in the name of
this universal basis.
In Ref.~\cite{CuiWang}, $R_{|2\rangle}$ is produced by injection of the magic state
\begin{equation} \label{eq:def:psi}
|\psi\rangle = |0\rangle - |1\rangle +|2\rangle.
\end{equation}
The injection circuit is coherent probabilistic, succeeds in three iterations on average and consumes three copies of the magic state $|\psi\rangle$ on average.

For completeness we present the logic of the injection circuit on Figure \ref{fig:R2:RUS}. Each directed arrow in the circuit is labeled with the result of standard measurement of the first qutrit in the state
$\mbox{SUM}_{2,1} \, (|\psi\rangle \otimes |input\rangle)$. On $m=0$ the sign of the third component of the input is flipped; on $m=1,2$ the sign of the first or second component respectively is flipped.

%%%%%%%%%%%%%%%%%%%%%%%%%%%%%%%%%%%%%%%%%%%%%%%%%%%%%%%%%%%%%%%%%
%
%  Figure: Markov chain corresponding to R2 gate
%
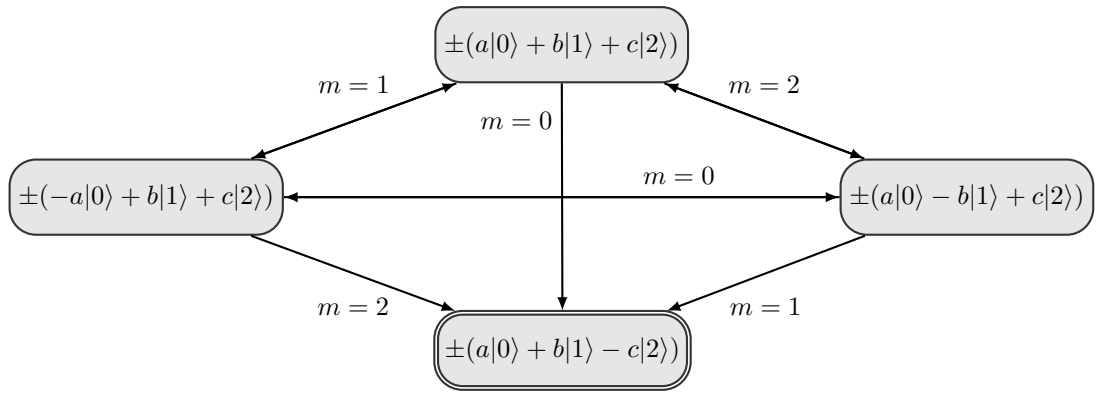
\begin{figure*}

\tikzstyle{state}=[rectangle,thick,rounded corners=10,minimum size=1cm,draw=black!80,fill=black!10]

\begin{tikzpicture}[>=latex,text height=1.5ex,text depth=0.25ex]
   % Nodes of the Markov chain
   \node (i_1) [state] {$\pm (a \ket{0} + b \ket{1} + c \ket{2})$};
   \node (i_2) [state,below left=10mm and 20mm of i_1] {$\pm (-a \ket{0} + b \ket{1} + c \ket{2})$};
     \node (i_3) [state,below right=10mm and 20mm of i_1] {$\pm (a \ket{0} - b \ket{1} + c \ket{2})$};
     \node (i_4) [state,accepting,below right=10mm and 20mm of i_2] {$\pm (a \ket{0} + b \ket{1} - c \ket{2})$};

   % Connecting the nodes via edges
  \path[->]
	(i_1) edge[thick] node [above=2mm] {$m=1$} (i_2)
	(i_2) edge[thick] node [above] {} (i_1)
	(i_1) edge[thick] node [above=2mm] {$m=2$} (i_3)
	(i_3) edge[thick] node [above] {} (i_1)
	(i_1) edge[thick] node [left]  {\raisebox{1cm}{$m=0$}} (i_4)
	(i_2) edge[thick] node [below=2mm] {$m=2$} (i_4)
	(i_3) edge[thick] node [below=2mm] {$m=1$} (i_4)
	(i_2) edge[thick] node [below] {} (i_3)
	(i_3) edge[thick] node [above] {\hspace*{3cm} $m=0$} (i_2);
\end{tikzpicture}
\caption{\label{fig:R2:RUS} Markov chain for repeat-until-success implementation of the injection of the $R_{|2\rangle}$ gate \cite{CuiWang}. Starting point is a general input $a\ket{0}+b\ket{1}+c\ket{2}$,
where $a,b,c \in \CC$. Arrows indicate transitions between single-qutrit states. Each arrow represent a single trial including measurement and consumption of the resource state $\ket{\psi}$, where each of the
transitions is labeled with the measurement result. The absorbing state corresponds to successful implementation of the $R_{|2\rangle}$ gate and is denoted by double borders.}
\end{figure*}

In the original anyonic framework the $|\psi\rangle$ state is produced by a relatively inexpensive protocol that uses topological measurement and consequent intra-qutrit projection (see \cite{CuiWang}, Lemma
5). This protocol requires only three qutrits and produces an exact copy of $|\psi\rangle$ in $9/4$ trials on average. This is much better than any state distillation method, especially because it produces a
copy of $|\psi\rangle$ with fidelity $1$.

In \cite{BCKW} we have developed effective compilation methods to compile efficient circuits in the metaplectic basis Clifford+$R_{|2\rangle}$.
In particular, given an arbitrary two-level Householder reflection $r$ and a desired target precision $\varepsilon$, then $r$ is effectively approximated by a metaplectic circuit of $R$-count at most
$8\,\log_3(1/\varepsilon)+O(\log(\log(1/\varepsilon)))$, where $R$-count is the number of occurrences of non-Clifford axial reflections in the circuit.
This allows us to approximate the $\mbox{CNOT}$ and Toffoli gates very tightly and at low cost over the metaplectic basis (see Section \ref{subsec:reflection:metaplectic}).
Moreover if we wanted constant-depth high-fidelity widgets for $\mbox{CNOT}$ and Toffoli we can do so by emulating, rather than distilling the magic state $|\mu\rangle$ of (\ref{eq:def:mu}) by a metaplectic
circuit and thus obtain a high fidelity emulation of the $P_9$ gate at constant online depth (see Section \ref{subsec:classical:with:P9}).

As we show in Appendix \ref{app:sec:R:2}, the converse also works. With available ancillas and enough reversible classical gates we can prepare the requisite magic state $|\psi\rangle$ exactly
 %the $R_{|2\rangle}$ gate
on a generic ternary computer. The particular method in the appendix is probabilistic circuit for the magic state $|\psi\rangle$ of (\ref{eq:def:psi}) using the classical non-Clifford gate
$C_2(\mbox{INC})$. Our current method for the latter gate is to implement it as a ancilla-free circuit with three $P_9$ gates.

\subsection{Top-level view of Shor's integer factorization algorithm} \label{subsec:shor:top:level}

The polynomial-time algorithm for integer factorization originally developed in Ref. \cite{Shor} is a hybrid algorithm that combines a quantum circuit with classical preprocessing and post-processing.
In general, the task of factoring an integer can be efficiently reduced classically to a set of hard cases. A hard case of the factorization problem comprises factoring a large integer $N$ that is odd, square-free
and composite.

 Let $a$ be a randomly picked integer that is relatively prime with $N$. By Euler's theorem, $a^{\phi(N)} = 1 \; {\rm mod} \; N$, where $\phi$ is the Euler's totient function,  and thus the modular
 exponentiation function $e_a: x \mapsto a^x \; {\rm mod} \; N$ is periodic with period $\phi(N) < N$.
Let now $0 < r < N$ be a period of the $e_a(x)$ function ($e_a(x+r)=e_a(x), \forall x$) and suppose, additionally that $r$ is even and $a^{r/2} \neq -1 \; {\rm mod} \; N$. Then the $\mbox{gcd}(a^{r/2}-1,N)$
must be a non-trivial divisor of $N$. The greatest common divisor is computed efficiently by classical means and it can be shown that the probability of satisfying the conditions $r = 0 \; {\rm mod} \; 2$ and
$a^{r/2} \neq -1 \; {\rm mod} \; N$ is rather high when $a$ is picked at random.
Therefore in Shor's algorithm a quantum circuit is only used for finding the small period $r$ of $e_a(x)$ once an appropriate $a$ has been randomly picked.

One quantum circuit to solve for $r$ consists of three stages:
\begin{enumerate}
\item Prepare quantum state proportional to the following superposition:
\begin{equation} \label{eq:Shor:superposition}
\sum_{k = 0}^{N^2} |k\rangle |a^k \; {\rm mod} \; N\rangle.
\end{equation}
\item Perform in-place quantum Fourier transform of the first register.
\item Measure the first register.
\end{enumerate}

The process is repeated until a classical integer state $j$ obtained as the result of measurement in step 3 enables recovery of a small period $r$ by efficient classical post-processing.
%We iterate steps 1) -- 3) until a useful $j$ is found and hence a value of $r$ is recovered.

Shor has shown  \cite{Shor} that the probability of successful recovery of $r$ in one of the iterations is in
$\Omega(1/\log(\log(N)))$. Therefore we will succeed ``almost certainly'' in finding a desired small period $r$ in $O(\log(\log(N)))$ trials.

Given the known efficiency of the quantum Fourier transform, most of the quantum complexity of this solution falls in 1, where the state (\ref{eq:Shor:superposition}) is prepared. Specific quantum circuits for
preparing this superposition have been proposed (cf. \cite{BeckmanEtAl,Beauregard,CleveWatrous,Markov2012,Markov2013,MeterItoh,TakaKuni,VedralEtAl,Zalka,Zalka2006}).

In the context of this paper, distinguish between two types of period-finding circuits. One type, as in Ref.~\cite{Beauregard}, is width-optimizing and uses approximate
arithmetic. These circuits interleave multiple quantum Fourier transform and inverse Fourier transform blocks into modular arithmetic circuits, which in practice leads to significant
depth overhead.  We forego the analysis of circuits of this type for the lack of space leaving such analysis for future research.

The second type are framed as exact reversible arithmetic circuits. Their efficient ternary emulation amounts to efficient emulation of $\mbox{CNOT}$ and Toffoli gates, possibly
after some peephole optimization.
We discuss two typical circuits of this kind in detail in Section \ref{sec:Emulation}, and briefly touch upon a number of alternatives in Appendix \ref{app:sec:alternative:circuits}.

It is important to note that, with a couple of exceptions
%(TODO),
the multi-qubit designs for Shor state preparation assumed ideal $\mbox{CNOT}$ and Toffoli gates.
However, in Clifford+T framework, for example, the Toffoli gate is often only as ideal as the $T$ gate.
The question of the required fidelity of $\mbox{CNOT}$ and Toffoli gates for the quantum period finding loop to work is an important one.

If the superposition (\ref{eq:Shor:superposition}) is prepared imperfectly, with fidelity $1-\varepsilon$ for some $\varepsilon$ in $o(1/\sqrt{\log(\log(N)})$,
then the probability of obtaining one of the ``useful'' measurements will be asymptotically the same as with the ideal superposition, i.e., in $\Omega(1/\log(\log(N)))$.  (For completeness, we spell out
the argument in Appendix \ref{app:sec:fidelity}.) Therefore, if $d$ is the depth of the corresponding quantum circuit preparing the state, then the bound on the required precision of the individual gates in
the circuit may be in
$o(1/(d\,\sqrt{\log(\log(N)}))$ in the context of Shor's algorithm.
%, which is a pretty sparing bound as we will see.

%TODO:MUST THREE-WAY COMPARISON NOW
In the rest of the paper we explore ternary emulations of binary period-finding circuits and compare them to truly ternary period-finding circuits with ternary encoding of integers. We demonstrate
that the fidelity and non-Clifford cost of such ternary circuits are reduced to those of the $C(\mbox{INC})$ gates. We also demonstrate that efficient emulation of binary period finding requires mostly binary
Toffoli gates with some use of $C(\mbox{INC})$.

\section{Multi-Qutrit and Multi-Qubit Arithmetic on Generic Ternary Quantum Computer} \label{sec:Emulation}

We explore two options for cost-efficient integer arithmetic over the ternary Clifford+$P_9$ basis: (a) by emulating arithmetic on binary-encoded data; and, (b) by performing arithmetic on
ternary-encoded data, based on tools developed in \cite{BCRS}.

Circuits for reversible ternary adders have been explored earlier. (See, for example, \cite{MillerEtAl,SatohEtAl,KhanPerkowski}). Since this field has been
in early stages so far, there is a lot of divergence in terminology: however in \cite{MillerEtAl,SatohEtAl,KhanPerkowski} the key non-Clifford tool
for the circuitry is an equivalent of the $C_f(\mbox{INC})$ gate, in our notation. As pointed out in \cite{BCRS}, our use of this tool is more efficient,
 mainly due to
the design of ``generalized'' carry gates and other reflection-based operations.

Our ternary circuits for emulated binary encoding of integers are new, as far as we know.

The emulated binary and genuine ternary versions of integer arithmetic have different practical bottlenecks, although they are asymptotically equivalent in terms of cost. With the ripple-carry
adders,  the emulated binary encoding wins, in practice, in both width and depth over the ternary encoding, whereas with carry-lookahead adders the ternary encoding achieves smaller width but
yields no notable non-Clifford depth advantage in the context of modular exponentiation.

To give the study a mathematical form, let us agree to take into account only non-Clifford gates used with either encoding and let us agree
to count a stack of non-Clifford gates performed in parallel in one time slice as a single \emph{unit of non-Clifford depth}.
We call the number of units of non-Clifford depth in a circuit the \emph{non-Clifford depth} of the circuit.

Throughout the rest of the paper we use the following

\begin{defin}
For integer $n>0$ let $|j\rangle$, $|k\rangle$ be two different standard basis vectors in the $n$-qudit Hilbert space. We call the classical gate
\begin{equation} \label{eq:def:tau}
\tau_{|j\rangle,|k\rangle} = I^{\otimes n} -  |j\rangle \langle j| -  |k\rangle \langle k| +  |j\rangle \langle k| + |k\rangle \langle j|
\end{equation}
a \emph{two-level axial reflection} in $n$ qudits.
\end{defin}

As a motivation for this term, note that $\tau_{|j\rangle,|k\rangle}$ can be rewritten as the two-level Householder reflection
\[
 I^{\otimes n} -  2\, |u\rangle \langle u|,  |u\rangle = (|j\rangle - |k\rangle)/\sqrt{2}.
\]
Clearly, in binary encoding, the $\mbox{CNOT}$, the Toffoli and any variably controlled Toffoli gate is a two-level axial reflection in the corresponding number of dimensions.

\subsection{Ternary circuit for binary ripple-carry additive shift} \label{subsec:ripple:carry}

We discuss emulating an additive shift circuit improving on a quantum ripple-carry adder from \cite{Cuccaro}. %The low-level resource counts that are used in subsequent sections
are cast in {\bf bold} font below.

\begin{comment}
That adder has simple structure and requires only two ancillas. In the context of Shor's algorithm it would be very practical to use the adder when factorizing medium-size numbers. However, since for $n$-bit
numbers the depth of the circuit is $O(n)$ , it is asymptotically inferior to lookahead additive shift described in the next subsection and may become unusable in practice when factoring very large numbers.
Either way, the simpler ripple carry adder gives an excellent opportunity to introduce the emulation concepts.
\end{comment}

Let $a$  be a classically known $n$-bit integer and $b$ be a quantumly-stored $n$-qubit basis state. We are looking for a quantum implementation of the function $|b\rangle \mapsto |a+b\rangle$.
More specifically, we are looking for a pre-compiled quantum circuit $C_a$ parameterized by $a$ which is known at compilation time.
Consider the well-known quantum ripple-carry adder from \cite{Cuccaro} (in particular, the circuit illustrated on Figure 4 for $n=6$ there that is copied, for completeness into our Fig.
\ref{fig:original:cuccaro}).

\begin{figure}[ht]
\input{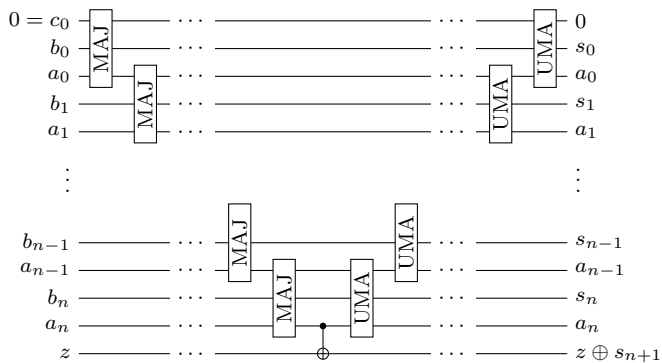}
\caption{\label{fig:original:cuccaro} Ripple-carry adder from \cite{Cuccaro}.}
\end{figure}

The adder uses $2n+2$ qubits. It performs a ladder of $n$ $\mbox{MAJ}$ gates to compute all the carry bits, including the top one. The top carry bit is copied unto the last qubit and
the ladder of $n$ $\mbox{UMA}$ gates is performed. Each $\mbox{UMA}$ gate uncomputes the corresponding  $\mbox{MAJ}$ function and performs the three-way $\mathbb{Z}_2$ addition
$a_i \oplus b_i \oplus c_i$.

It is somewhat hard to fold in the classically known $a$ in the multi-qubit framework using this design. Note however, that a solution along these lines is offered in \cite{TakaKuni}.
However it is easy to fold in $a$ in ternary emulation using the third basis state of the qutrit. We show that it takes exactly $n+2$ qutrits to emulate the binary shift
$|b\rangle \mapsto |a+b\rangle$.
 %and that by adding one ancillary qutrit to the circuit to the total of $n+3$ qutrits we can do the emulation very efficiently.

\begin{comment}
It is helpful for our purposes that the $\; {\rm mod} \; 2$ additive shifts $b_i \oplus$ and $a_i \oplus$ can be taken out of the $\mbox{UMA}$ blocks, pushed into the end of the circuit and all the additive
shifts can be done in parallel in depth at most two there.
\end{comment}

Consider $n+2$ qutrits where the top and the bottom ones are prepared in $|0\rangle$ state and the remaining $n$ encode the {\emph binary bits} of the  $|b\rangle$.

We will be looking for reversible two-qutrit gates $Y_0, Y_1$ such that
\begin{equation} \label{eq:for:Y:gate}
Y_{a_j} \, |c_j,b_j\rangle = |c'_j,c_{j+1}\rangle
\end{equation}
\noindent where $c_{j+1}$ is the correct carry bit for $c_j+a_j+b_j$ and $c'_j$ is an appropriate trit.

\begin{comment}
Consider the two two-qutrit gates parameterized by $\mathbb{Z}_2$ : $Y_0 = \tau_{|01\rangle,|20\rangle}$  and $Y_1 = \tau_{|10\rangle,|21\rangle}$.
Assuming $a$, $b$ and all the carries are encoded as binary numbers, it is easily verified by direct computation that
$Y_{a_j} \, |c_j,b_j\rangle = |c'_j,c_{j+1}\rangle$, where $c_{j+1}$ is the correct carry bit for the $a_j+b_j+c_j$ addition and $c'_j$ is a trit that can assume the value of $2$ in two out of the eight
cases.
\end{comment}

Since all the bits of $a$ are known we can precompile a ladder of $Y$ gates that correctly computes the top carry bit $c_n$ and puts the modified carry trit $c'_j$ on each $b_j$ wire. Having copied $c_n$ onto
the last qutrit, we sequentially undo the $Y$ gates in lockstep with computing partial $\mathbb{Z}_2$-sums $b_j\oplus c_j$ on all the $b_j$ wires using gates of $\mbox{CNOT}$ type.

We note that $Y_0,Y_1$ are ternary gates, used however in a narrow context of a truth table with just four columns. One would intuitively expect that their restriction to the context can be emulated
at a relatively small expense.

Indeed:

\begin{prop} \label{prop:Ys:improved}
Label the $c_i$ wire by $0$ and $b_i$ wire by $1$
In the context of binary data the gates
\[
\begin{split}
Y_0 = C_2(\mbox{INC})_{0,1}^{\dagger} \mbox{SUM}_{1,0} (\tau_{|0\rangle,|1\rangle} \otimes I)
\end{split}
\]
\noindent and
\[
\begin{split}
Y_1  = C_2(\mbox{INC})_{0,1} (I \otimes \tau_{|0\rangle,|1\rangle} ) \mbox{SUM}_{1,0} (I \otimes \tau_{|0\rangle,|1\rangle} )
\end{split}
\]
\noindent satisfy the condition (\ref{eq:for:Y:gate}).

Here the $C_2(\mbox{INC})$ is the binary-controlled increment
$C_2(\mbox{INC}): |j,k\rangle \mapsto |j,k+\delta_{j,2}\rangle$

\end{prop}
\begin{proof}
By direct computation. Note, that we do not care, what either of these two circuits does outside of the binary data subspace as long as the action is reversible.
\end{proof}

The $C_2(\mbox{INC})$ gate is also denoted as $C_2(X)$ in Ref. \cite{BCRS}, where its cost and utility is discussed in detail (see also further discussion in section \ref{sec:reflections}). The non-Clifford cost of
either $Y_j$ gate is equal to the non-Clifford cost of $C_2(\mbox{INC})$ which is known to be $3$ $P_9$ gates. Allowing one ancillary qutrit, the $C_2(\mbox{INC})$ is represented by a circuit of $P_9$-depth of
$1$ and $P_9$-width of $3$.

Besides the generalized carry computation, the additive shift circuit also needs to perform the bitwise $\; \mod 2\;$ addition by emulated gates of $\mbox{CNOT}$ type.

Recall that $\mbox{CNOT}$  gate cannot be exactly represented by a ternary Clifford circuit (cf. \cite{BCRS}, Appendix A). As shown further in Proposition \ref{prop:CNOT:ancilla-free}, the non-Clifford cost of
ternary-emulated $\mbox{CNOT}$ on binary data only is an equivalent of two $C_2(\mbox{INC})$.
% gates or $6$ $P_9$ gates.
Thus the additive shift takes roughly ${12\, n}$ $P_9$ gates to complete (not counting the Clifford scaffolding). With one ancilla this can be done at $P_9$-depth of $4\, n$ and $P_9$-width of $3$.

However, Shor's period funding functions relies on controlled and doubly-controlled versions of the additive shift. It suffices to control only the bitwise addition gates. Thus adding one level of control
produces $n$ additional Toffoli gates and adding the second level of control turns these gates into controlled Toffolis.
This is the bottleneck of the emulated solution: as per Corollaries \ref{corol:Toffoli:the:12} and \ref{corol:ctrl:Toffoli:18} in section below, an emulated Toffoli takes $12$ $P_9$ and the
binary-controlled Toffoli takes $18$ $P_9$ gates respectively.

Thus overall the controlled shift takes ${18 \, n}$ $P_9$  and the doubly-controlled shift takes ${24 \, n}$ $P_9$ gates. Allowing, again, an ancillary qutrit the $P_9$-depths of the corresponding
circuits can be made $6\,n$ and $8\,n$ respectively.

For what it is worth, the $P_9$-counts in this solution are similar (and in fact marginally lower) than the $T$-counts required for running the original binary adder  \cite{Cuccaro} on the more common binary
Clifford+T platform. Indeed each of the $\mbox{MAJ}$ and $\mbox{UMA}$ gates shown on figure \ref{fig:original:cuccaro} is Clifford-equivalent to a Toffoli gate that takes $7$ $T$ gates to implement. Adding one
level of control to the adder increases the non-Clifford complexity by an additional $n$ Toffoli gates to the total $T$-count of $21\,n$. Adding the second level of control, conservatively, brings in $2\,n$
additional modified Toffoli gates to yield the total $T$-count of $29\,n$.

We also note that the width of the ternary emulation circuit is equal to $n+2$ qutrits, whereas the original purely binary design appears to require $2\, n + 2$ qubits.

The construction of Corollaries \ref{corol:Toffoli:the:12} and \ref{corol:ctrl:Toffoli:18} requires 1 and 2 ancillas respectively. These ancillae can be shared along the depth of the circuit inflating the overall
width by 2 qutrits.

\subsection{Ternary circuit for ternary ripple-carry additive shift} \label{subsec:ternary:ternary:ternary}
Consider ripple-carry implementation of the quantum function $|b\rangle \mapsto |a+b\rangle$, where $|b\rangle$ is quantumly encoded integer and $a$ is an integer that is classically known.
Suppose $a$ and $b$ are encoded as either bit strings with at most $n$ bits or as trit strings with at most $m=\lceil \log_3(2) n \rceil$ trits (with $\log_3(2) \approx 0.63093$).
Since $a$ is classically known, we strive to improve on the ternary ripple-carry adder of \cite{BCRS} by folding in the trits of $a$. However, we are no longer able to encode all of the quantum information for
$b$ and the carry on the same qutrit. The additive shift thus requires  roughly $2\,m-w_1(a)$  qutrits to run (where $w_1(a)$ is the number of trits equal to $1$ in the ternary expansion of $a$).

The ternary additive shift in this design has somewhat higher non-Clifford time cost compared to the emulated binary shift of section \ref{subsec:ripple:carry}.

%The first observation follows from an analysis of the structure of generalized carry gates.
For the classical additive shift we do not physically encode the trits of $a$ and instead pre-compile different
generalized carry circuits for different values of these trits. Tables \ref{table:ai:1} and \ref{table:ai:0} show the truth tables for the consecutive carry $c_{i+1}$ given, respectively, $a_i = 1$ and $a_i=0$
(the case of $a_i=2$ is symmetric to the case $a_0$ and yields the came conclusions).

\begin{table*}
    \begin{tabular}{l@{\qquad\qquad}c@{\qquad\qquad}c@{\qquad\qquad}c@{\qquad\qquad}c@{\qquad\qquad}c@{\qquad\qquad}c} \hline\hline\\[-1.5ex]
    $c_i$                                  & $0$                         & $0$ & $0$                         & $1$ & $1$                         & $1$\\[1ex]
    $b_i$        & $0$                         & $1$ & $2$                         & $0$ & $1$                         & $2$  \\
    $c_{i+1}$        & $0$                         & $0$ & $1$                         & $0$ & $1$                         & $1$\\[1ex]
\hline\hline
    \end{tabular}
    \caption{Truth table for $c_{i+1}$ given $a_i = 1$} \label{table:ai:1}
  \end{table*}

\begin{table*}
    \begin{tabular}{l@{\qquad\qquad}c@{\qquad\qquad}c@{\qquad\qquad}c@{\qquad\qquad}c@{\qquad\qquad}c@{\qquad\qquad}c} \hline\hline\\[-1.5ex]
    $c_i$                                  & $0$                         & $0$ & $0$                         & $1$ & $1$                         & $1$\\[1ex]
    $b_i$        & $0$                         & $1$ & $2$                         & $0$ & $1$                         & $2$  \\
    $c_{i+1}$        & $0$                         & $0$ & $0$                         & $0$ & $0$                         & $1$\\[1ex]
\hline\hline
    \end{tabular}
    \caption{Truth table for $c_{i+1}$ given $a_i = 0$} \label{table:ai:0}
  \end{table*}

The case of $a_i = 1$ does not require any ancillary qutrits since the $c_{i+1}$ is a balanced binary function that can be produced reversibly on the pair of qutrits encoding $c_i$ and $b_i$ by ternary
$\mbox{SWAP}$ gate followed by $|01\rangle \leftrightarrow |20\rangle$.

However, in the case of $a_i = 0$ the $c_{i+1}=0$ in five cases (respectively, for $a_i = 2$ the $c_{i+1}=1$ in five cases) and such five basis vectors cannot be represented in two-qutrit state space. These cases
thus require an ancillary qutrit to encode $c_{i+1}$.

In the case of $a_i = 0$ we simply take the ancilla in the $|0\rangle$ state and apply doubly-controlled $\mbox{INC}$ gate with the ternary control on $c_i$ and binary control on $b_i$. In the case of $a_i = 2$ it
suffices to additionally use the Clifford $\tau_{|0\rangle,|1\rangle}$ gate on the $c_{i+1}$.

Assuming $a$ is generic with $w_1(a) \approx m/3$ we get an average width of the additive shift circuit of roughly $5/3 \, m$ which eliminates the space savings afforded by denser ternary
encoding ($5/3 \log_3(2) \approx 1.05$).

Let us now make case for the second observation.
%While the width counting above is almost independent on the specific gate syntheses, the second observation relies on specific results of \cite{BCRS} (that are the best known to us, but haven't been yet
proven practically optimal).
We start by assessing the clean magic state counts for simple uncontrolled additive shift. We note that for any classical value of the $a_i$ trit the non-Clifford cost of the carry gate is the same and equals
$15$ clean magic states. Indeed, depending on $a_i$ and in terminology of \cite{BCRS} we either need one gate of $S_{01,10}$ type or one gate of $C_0(\mbox{SUM})$ type. In subsection 5.1 of the \cite{BCRS}
both types are reduced to $5$ binary-controlled increments and consequently to $15$ $P_9$ gates. The concluding trit-wise addition is done by Clifford $\mbox{SUM}$s at negligible cost. Thus the overall cost of
the circuit is roughly  ${30\,m \approx 19 \, n}$ $P_9$ gates. Allowing an ancillary qutrit, the $P_9$-depth of the circuit can be made equal to ${10\,m > 6 \, n}$.

Adding one ternary control to the circuit turns all the finalizing $\mbox{SUM}$s into``Horner'' gates $\Lambda(\mbox{SUM})$ that overall takes $4\, m$ additional $P_9$ gates to the total non-Clifford cost of
${34\,m > 21 \, n}$ $P_9$ gates.

A subtle point discussed in section \ref{subsec:mod:exponetiation} below is that the second control that is routinely added to the additive shift gate $S_a$ is in fact strict control that turns it into a
$C_f(S_a)$ gate $f\in \{1,2\}$ where $S_f$ is activated only by the control basis state $|f\rangle$. This turns each of the the $m$ ``Horner'' gates into a four-qutrit  $C_f(\Lambda(\mbox{SUM}))$ gate. We do
not have an available ancilla-free design for a synthesis of this gate.

Our best design described in Proposition \ref{prop:C:Lambda:SUM} sets the non-Clifford cost at $23$ $P_9$ gates given one clean ancilla.

Thus adding the required second (strict) control inflates the overall cost of the ternary circuit to ${53 \, m > 33 \,n}$ $P_9$ gates.

Again, with available ancilla the circuit can be restacked to $P_9$-width of $3$ reducing the $P_9$-depth by the factor of $3$ (to roughly ${19 \, m}$ in case of doubly-controlled additive shift).
%Adding another ternary control turns each "Horner" gate into $\Lambda \Lambda(\mbox{SUM})$ costing $12$ $P_9$ gates and the total non-Clifford cost rises to  ${42\,m \approx 26.5 \, n}$ $P_9$ gates. This
is less than the non-Clifford cost of the emulated binary $n$-bits doubly-controlled additive shift.

The comparative cost of the binary and ternary options is summarized in table \ref{table:binary:vs:ternary}.

\begin{table*}
    \begin{tabular}{l@{\qquad\qquad}l@{\qquad\qquad}l} \hline\hline\\[-1.5ex]
\multicolumn{1}{c}{\hspace*{-2cm} Circuits}
           &  \multicolumn{1}{l}{$\# P_9$: emulated binary} &  \multicolumn{1}{l}{$\# P_9$: ternary}\\[0.5ex]
    \hline\\[-1.5ex]
   Simple additive shift                                 & $12\,n$                       & $19 \,n$ \\[1ex]
    Controlled additive shift       & $18\,n$    & $> 21\,n $  \\[1ex]
    Doubly-controlled additive shifts        & $24\,n$   & $> 33 \, n$ \\[1ex]
\hline\hline
    \end{tabular}
    \caption{Cost of ripple-carry additive shift: ternary vs. emulated binary. $n$ is the bit size of the arguments.} \label{table:binary:vs:ternary}
  \end{table*}

We demonstrate in the Section \ref{subsec:modular:shifts} that the best-known ternary-controlled modular shift circuit requires $4$ instead of $3$ additive shift blocks on roughly half of the
modular addition cases, so, in the context of the required modular addition, the emulated binary encoding appears to be a practical win-win when a low width ripple-carry adder is used.

\subsection{Circuits for carry lookahead additive shift} \label{subsec:carry:lookahead}
The resource layout is different for known carry lookahead solutions. For the sake of space we forego detailed analysis and only sketch the big picture.

We assume, that the integers $a$ and $b$ are encoded as either bit strings with at most $n$ bits or as trit strings with at most $m=\lceil \log_3(2) n \rceil$ trits.
We use carry lookahead additive shifts based on the in-place multi-qubit carry lookahead adder \cite{DraperSvore} and the in-place multi-qutrit carry lookahead adder \cite{BCRS}.

The non-Clifford depths of the corresponding circuits are $4\, \log_2(n)$ and $4\, \log_2(m)$  respectively up to small additive constants.

Because $\log_2(m) = \log_2(n) + \log_2(\log_3(2))$, there is no substantial difference in non-Clifford depths. The non-Clifford layers of the binary adder are populated with Toffoli gates and for the ternary
adder they are populated with carry status merge/unmerge widgets (the $\mathcal{M}$ and $\mathcal{M}^{\dagger}$ widgets of \cite{BCRS}). The cost of ancilla-free emulation of the former or,
respectively, execution of the latter is identical with $15$ $P_9$ gates.

When levels of control are added to the shift circuit, putting ternary control on ternary widgets is more expensive than building multi-controlled Toffoli gates, as discussion in Section
\ref{subsec:ripple:carry} implies. But in the context of carry lookahead circuits the multi-controlled gates are located in just two layers out of $O(\log(n))$ thus the impact of this cost distinction is both
asymptotically and practically negligible.

Note that the widths of the binary and ternary circuits are roughly proportional to $n$ and $m=\lceil \log_3(2) n \rceil$, respectively. This means that the purely ternary solution has roughly $m/n \approx \log_3(2)$ smaller
width.

%To summarize, in the context of the required modular addition, the choice between emulated binary encoding and purely ternary encoding of the carry lookahead addition is a width to depth tradeoff. Since the
depth overhead percentage is moderate, we should prefer purely ternary encoding when implementing Shor's period finding on small quantum computer.

\subsection{Circuits for modular additive shifts} \label{subsec:modular:shifts}

We review layout for modular additive shift and controlled modular additive shift in both emulated binary and genuine ternary setups.

Let $N >> 0$ and $a < N$ be classically known integers.
The commonly used scheme to compute the quantum modular additive shift $|b\rangle \mapsto |(a+b) \; {\rm mod} \; N\rangle$ is to compute $|a+b\rangle$, figure out whether $a+b <N$ and, if not, then subtract
$N$.
In order to do it coherently without measurement we need to
\begin{enumerate}
\item Speculatively compute the $|(a-N)+b\rangle$ shift; structure it so that the top carry bit $c_{n+1}$ is $1$ iff $(a-N)+b <0$.
\item Copy $c_{n+1}$ to a clean ancilla $x$.
\item Apply the shift by $+N$ controlled by the ancilla $x$.
\item Clean up the ancilla $x$.
\end{enumerate}

Surprisingly, the last step is less than trivial.
We need to compare the encoded integer $|y\rangle$ after step 3) to $a$. Then $y \geq a$ if and only if $c_{n+1}=1$. Therefore we must flip the ancilla if and only if $y \geq a$. We do this by taking a circuit
for comparison to classical threshold and wiring the $\mbox{NOT}\, x$ into it in place of the top carry qubit. It is easy to see that performing the comparison circuit has the exactly the desired effect on the
ancilla $x$.
A top level layout of the modular additive shift is shown in Figure \ref{fig:modular:shift}. We note that the three-stage layout shown in the Figure is not entirely new. It is very similar to designs proposed
in \cite{MeterItoh} and \cite{TakaKuni}.
Clearly the non-Clifford depth of this scheme is roughly triple the non-Clifford depth of the additive shift circuit in either binary or ternary framework.

\begin{figure*}%[ht]
%\includegraphics[width=6in]{NewModularShift.pdf}
%AB:PLEASE DON'T TOUCH THE SCALING FACTOR
\begin{tikzpicture}[scale=3.50000,x=1pt,y=1pt]
\filldraw[color=white] (0.000000, -7.500000) rectangle (90.000000, 37.500000);
% Drawing wires
% Line 2: a W |b\rangle
\draw[color=black] (0.000000,30.000000) -- (90.000000,30.000000);
\draw[color=black] (0.000000,30.000000) node[left] {$|b\rangle$};
% Line 3: b W |0\rangle
%\draw[color=black] (0.000000,15.000000) -- (90.000000,15.000000);
\draw[color=black] (0.000000,15.000000) -- (71.000000,15.000000);
\draw[color=black] (85.000000,15.000000) -- (90.000000,15.000000);
\draw[color=black] (0.000000,15.000000) node[left] {$|0\rangle$};
% Line 4: c W |0\rangle
\draw[color=black] (0.000000,0.000000) -- (90.000000,0.000000);
\draw[color=black] (0.000000,0.000000) node[left] {$|0\rangle$};
% Done with wires; drawing gates
% Line 5: a b G $+(a-N)$
\draw (12.000000,30.000000) -- (12.000000,15.000000);
\begin{scope}
\draw[fill=white] (12.000000, 22.500000) +(-45.000000:8.485281pt and 19.091883pt) -- +(45.000000:8.485281pt and 19.091883pt) -- +(135.000000:8.485281pt and 19.091883pt) -- +(225.000000:8.485281pt and
19.091883pt) -- cycle;
\clip (12.000000, 22.500000) +(-45.000000:8.485281pt and 19.091883pt) -- +(45.000000:8.485281pt and 19.091883pt) -- +(135.000000:8.485281pt and 19.091883pt) -- +(225.000000:8.485281pt and 19.091883pt) --
cycle;
\draw (12.000000, 22.500000) node {$+(a-N)$};
\end{scope}
% Line 6: b +c
\draw (33.000000,15.000000) -- (33.000000,0.000000);
\filldraw (33.000000, 15.000000) circle(1.000000pt);
\begin{scope}
\draw[fill=white] (33.000000, 0.000000) circle(3.000000pt);
\clip (33.000000, 0.000000) circle(3.000000pt);
\draw (30.000000, 0.000000) -- (36.000000, 0.000000);
\draw (33.000000, -3.000000) -- (33.000000, 3.000000);
\end{scope}
% Line 7: a b G $+N$ c
\draw (54.000000,30.000000) -- (54.000000,0.000000);
\begin{scope}
\draw[fill=white] (54.000000, 22.500000) +(-45.000000:8.485281pt and 19.091883pt) -- +(45.000000:8.485281pt and 19.091883pt) -- +(135.000000:8.485281pt and 19.091883pt) -- +(225.000000:8.485281pt and
19.091883pt) -- cycle;
\clip (54.000000, 22.500000) +(-45.000000:8.485281pt and 19.091883pt) -- +(45.000000:8.485281pt and 19.091883pt) -- +(135.000000:8.485281pt and 19.091883pt) -- +(225.000000:8.485281pt and 19.091883pt) --
cycle;
\draw (54.000000, 22.500000) node {$+N$};
\end{scope}
\filldraw (54.000000, 0.000000) circle(1.000000pt);
% Line 8: a b c G $\geq a ?$
\draw (78.000000,30.000000) -- (78.000000,0.000000);
\begin{scope}
\draw[fill=white] (78.000000, 15.000000) +(-45.000000:8.485281pt and 29.698485pt) -- +(45.000000:8.485281pt and 29.698485pt) -- +(135.000000:8.485281pt and 29.698485pt) -- +(225.000000:8.485281pt and
29.698485pt) -- cycle;
\clip (78.000000, 15.000000) +(-45.000000:8.485281pt and 29.698485pt) -- +(45.000000:8.485281pt and 29.698485pt) -- +(135.000000:8.485281pt and 29.698485pt) -- +(225.000000:8.485281pt and 29.698485pt) --
cycle;
\draw (78.000000, 15.000000) node {$\geq a ?$};
\end{scope}
% Done with gates; drawing ending labels
% Done with ending labels; drawing cut lines and comments
% Done with comments
\end{tikzpicture}
\caption{\label{fig:modular:shift} Top-level layout of modular additive shift for binary encoding. }
\end{figure*}
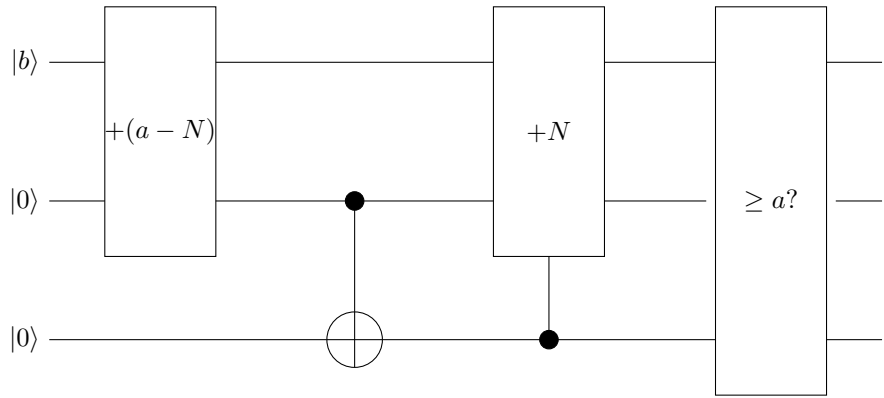

In the context of ternary encoding of integers and allowing for ternary control the logic turns out to be more  involved.
Depending on whether $2 \, a < N$ or not, which is known at compilation time, we need to compile two different circuits.
When $2 \, a < N$ we need to speculatively precompute $b+ c\, a -N$ where $c$ is the quantum value of the control trit. This is different from adding ternary control to the additive shift $+(a-N)$. A
straightforward way to do this is by taking the controlled shift $+c\, (a-N)$ followed by strictly controlled shift $C_2(+N)$. Aside from this additional shift box,  the circuit in Figure
\ref{fig:modular:shift} still works as intended, which is easy to establish: the speculative $b+ c\, a -N$ is corrected back to $b+ c\, a$ if and only if the eventual result is $\geq c\,a$ which is the
condition for the ancilla cleanup.

When $ 2\, a > N$ we can precompile ternary control on the entire $+(a-N)$ box, which then precomputes the
$y=b+ c (a -N)$ for us. However, here we still get some overhead compared to the binary encoding context. Indeed, we need to correct the speculative state $y$ to $y = b + c (a -N) +N$ when $y<0$ and it is
easily seen that the result is $\geq c (a -N) +N$ if and only if $y$ was negative and the correction happened. Thus the ancilla cleanup threshold is $t = c (a -N) +N$ on this branch. Since $c$ is the quantum
control trit, the comparison to $t$ is somewhat more expensive to engineer than comparison to $c\,a$.

To summarize, a purely ternary modular shift circuit allowing for ternary control would be similar to one shown in Figure \ref{fig:modular:shift:ternary}, where the extra dashed $C_2(+N)$ box is inserted at
compilation time when $2 \, a < N$. The latter case constitutes the critical path where we have to use an equivalent of $4$ additive shifts instead of $3$.

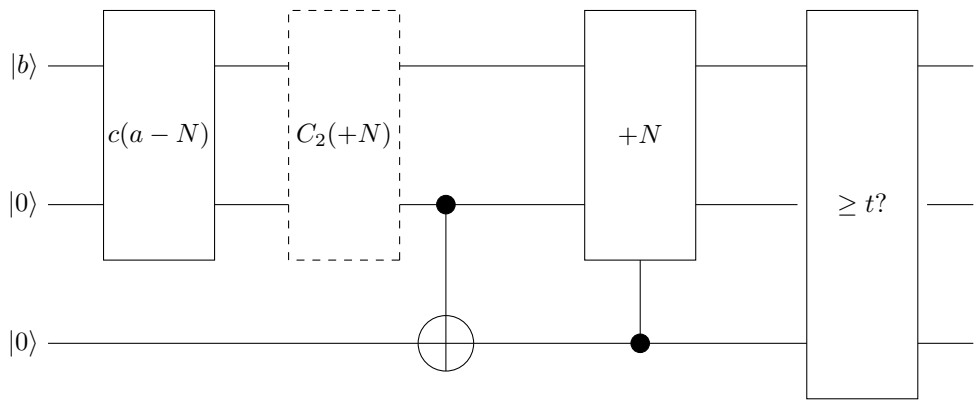
\begin{figure*}%[ht]
%\includegraphics[width=6in]{NewModularShift.pdf}
%AB:PLEASE DON'T TOUCH THE SCALING FACTOR
\begin{tikzpicture}[scale=3.50000,x=1pt,y=1pt]
\filldraw[color=white] (0.000000, -7.500000) rectangle (90.000000, 37.500000);
% Drawing wires
% Line 2: a W |b\rangle
\draw[color=black] (0.000000,30.000000) -- (100.000000,30.000000);
\draw[color=black] (0.000000,30.000000) node[left] {$|b\rangle$};
% Line 3: b W |0\rangle
%\draw[color=black] (0.000000,15.000000) -- (90.000000,15.000000);
\draw[color=black] (0.000000,15.000000) -- (81.000000,15.000000);
\draw[color=black] (95.000000,15.000000) -- (100.000000,15.000000);
\draw[color=black] (0.000000,15.000000) node[left] {$|0\rangle$};
% Line 4: c W |0\rangle
\draw[color=black] (0.000000,0.000000) -- (100.000000,0.000000);
\draw[color=black] (0.000000,0.000000) node[left] {$|0\rangle$};
% Done with wires; drawing gates
% Line 5: a b G $+(a-N)$
\draw (12.000000,30.000000) -- (12.000000,15.000000);
\begin{scope}
\draw[fill=white] (12.000000, 22.500000) +(-45.000000:8.485281pt and 19.091883pt) -- +(45.000000:8.485281pt and 19.091883pt) -- +(135.000000:8.485281pt and 19.091883pt) -- +(225.000000:8.485281pt and
19.091883pt) -- cycle;
\clip (12.000000, 22.500000) +(-45.000000:8.485281pt and 19.091883pt) -- +(45.000000:8.485281pt and 19.091883pt) -- +(135.000000:8.485281pt and 19.091883pt) -- +(225.000000:8.485281pt and 19.091883pt) --
cycle;
\draw (12.000000, 22.500000) node {$c(a-N)$};
\end{scope}
% Manual
\draw[dashed] (32.000000,30.000000) -- (32.000000,15.000000);
\begin{scope}
\draw[dashed][fill=white] (32.000000, 22.500000) +(-45.000000:8.485281pt and 19.091883pt) -- +(45.000000:8.485281pt and 19.091883pt) -- +(135.000000:8.485281pt and 19.091883pt) -- +(225.000000:8.485281pt and
19.091883pt) -- cycle;
\clip (32.000000, 22.500000) +(-45.000000:8.485281pt and 19.091883pt) -- +(45.000000:8.485281pt and 19.091883pt) -- +(135.000000:8.485281pt and 19.091883pt) -- +(225.000000:8.485281pt and 19.091883pt) --
cycle;
\draw (32.000000, 22.500000) node {$C_2(+N)$};
\end{scope}

% Line 6: b +c
\draw (43.000000,15.000000) -- (43.000000,0.000000);
\filldraw (43.000000, 15.000000) circle(1.000000pt);
\begin{scope}
\draw[fill=white] (43.000000, 0.000000) circle(3.000000pt);
\clip (43.000000, 0.000000) circle(3.000000pt);
\draw (40.000000, 0.000000) -- (46.000000, 0.000000);
\draw (43.000000, -3.000000) -- (43.000000, 3.000000);
\end{scope}
% Line 7: a b G $+N$ c
\draw (64.000000,30.000000) -- (64.000000,0.000000);
\begin{scope}
\draw[fill=white] (64.000000, 22.500000) +(-45.000000:8.485281pt and 19.091883pt) -- +(45.000000:8.485281pt and 19.091883pt) -- +(135.000000:8.485281pt and 19.091883pt) -- +(225.000000:8.485281pt and
19.091883pt) -- cycle;
\clip (64.000000, 22.500000) +(-45.000000:8.485281pt and 19.091883pt) -- +(45.000000:8.485281pt and 19.091883pt) -- +(135.000000:8.485281pt and 19.091883pt) -- +(225.000000:8.485281pt and 19.091883pt) --
cycle;
\draw (64.000000, 22.500000) node {$+N$};
\end{scope}
\filldraw (64.000000, 0.000000) circle(1.000000pt);
% Line 8: a b c G $\geq a ?$
\draw (88.000000,30.000000) -- (88.000000,0.000000);
\begin{scope}
\draw[fill=white] (88.000000, 15.000000) +(-45.000000:8.485281pt and 29.698485pt) -- +(45.000000:8.485281pt and 29.698485pt) -- +(135.000000:8.485281pt and 29.698485pt) -- +(225.000000:8.485281pt and
29.698485pt) -- cycle;
\clip (88.000000, 15.000000) +(-45.000000:8.485281pt and 29.698485pt) -- +(45.000000:8.485281pt and 29.698485pt) -- +(135.000000:8.485281pt and 29.698485pt) -- +(225.000000:8.485281pt and 29.698485pt) --
cycle;
\draw (88.000000, 15.000000) node {$\geq t?$};
\end{scope}
% Done with gates; drawing ending labels
% Done with ending labels; drawing cut lines and comments
% Done with comments
\end{tikzpicture}
\caption{\label{fig:modular:shift:ternary} Top-level layout of ternary modular additive shift. In case $2\,a < N$ the circuit is compiled with the additional $C_2(+N)$ shift controlled on $c=2$ and using the
threshold $t=c\, a$. In case $2\, a > N$ the additional shift is not needed, but the threshold $t=c(a-N) + N$.}
\end{figure*}

\begin{comment}
In ternary framework we propose that each of the three components of the modular shift circuit described above emulates the corresponding binary component as described in the preceding sections. Given this
design, the non-Clifford depth overhead factor due to ternary emulation (compared to binary solution) is asymptotically insignificant when $n \rightarrow \infty$ and in practice does not exceed
$(1+3/(4\,\log_2(n)))$.
The same applies to the controlled versions of the modular shift circuit.
\end{comment}

\subsection{Circuits for modular exponentiation} \label{subsec:mod:exponetiation}

For modular exponentiation $|k\rangle |1\rangle \mapsto |k\rangle |a^k \; {\rm mod} \; N\rangle$ we follow the known implementation proposed in the first half of Ref. \cite{Zalka}. Our designs are also
motivated in part by Ref. \cite{CleveWatrous}.

We denote by $d$ the dimension of the single qudit. $d$ is assumed to be either $2$ or $3$ where it matters.
Suppose that $a,N$ are classically known integers $a < N$, and $n$ is an integer approximately equal to $\log_d(N)$.

Suppose $|k\rangle$ is quantumly encoded, $k=\sum_{j=0}^{2 \, n-1} k_j \, d^j$ is base-$d$ expansion of $k$, where $k_j$ are the corresponding qudit states.
First, we observe that
\begin{equation} \label{eq:mod:exp:product}
a^k \; {\rm mod} \; N = \prod_{j=0}^{2\,n-1} (a^{d^j} \; {\rm mod} \; N)^{k_j}  \; {\rm mod} \; N.
\end{equation}
Note that $(a^{d^j} \; {\rm mod} \; N)$ are $2\,n$ classical values that are known and easily pre-computable at compilation time.
Thus $|a^k \; {\rm mod} \; N \rangle$ is computed as a sequence of modular multiplicative shifts, each quantumly controlled by the $|k_j\rangle$.

Suppose we have computed the partial product
\[
p_{k,m} = \prod_{j=0}^{m} (a^{d^j} \; {\rm mod} \; N)^{k_j}  \; {\rm mod} \; N,
\]
and let
\[
p_{k,m} = \sum_{\ell=0}^{n-1} p_{k,m,\ell} d^{\ell}
\]
be the base-$d$ expansion of $p_{k,m}$.
Then
\[
p_{k,m+1} = \sum_{\ell=0}^{n-1} p_{k,m,\ell} (d^{\ell} a^{d^{m+1}} \; {\rm mod} \; N)^{k_{m+1}} \; {\rm mod} \; N.
\]

Observe, again, that

\begin{equation} \label{eq:d:ary:shifts}
\{(d^{\ell} a^{d^{m+1}} \; {\rm mod} \; N)^f \; {\rm mod} \; N | f \in [1..d-1]\}
\end{equation}

\noindent is the set of fewer than $d$ pre-computable classical values known a priori. Therefore, promoting $p_{k,m}$ to $p_{k,m+1}$ is performed as a sequence of modular additive shifts, controlled by
$p_{k,m,\ell}$ and $k_{m+1}$.

Herein lies a subtle difference between the case of $d=2$ and the case of $d>2$ (e.g. $d=3$).
In the case of $d=2$ we do the modular shift by $2^{\ell} a^{2^{m+1}} \; {\rm mod} \; N$ if and only if $p_{k,m,\ell} = k_{m+1} = 1$. Thus the corresponding gate is simply the doubly-controlled modular
additive shift.

In case of $d>2$ the $d-1$ basis values of $k_{m+1}$ lead to modular additive shift by one of the $d-1$ potentially different values listed in the equation (\ref{eq:d:ary:shifts}). Thus we need a $(d-1)$-way
quantum switch capable of selection between the listed values. Let $S_f,f \in [1..d-1]$ be the modular additive shift by the $f$-th value in (\ref{eq:d:ary:shifts}). Then the desired switch can be realized
coherently as the product $C_1(S_1) \cdots C_{d-1}(S_{d-1})$ where $C_f(S_f)$ is the $S_f$ activated only by the basis state $k_{m+1} = |f\rangle$.

This implies the following difference in the circuit makeup between the case of $d=2$ and the case of $d=3$.

For $d=2$ modular exponentiation takes roughly $2\, n^2$ doubly-controlled modular additive shifts; for $d=3$ it takes roughly $4\, m^2$ doubly-controlled modular additive shifts (where $m$ is the trit size of
the arguments), each with one ternary and one strict control on one of the two ternary values.

When comparing the option of performing the circuit in emulated binary encoding against the option of running it in true ternary encoding we find a practical dead heat between the two options in terms of
circuit depth. Indeed in counting the number of doubly-controlled additive shift boxes we find that $4\, m^2 = 2\,(\log_3(2))^2 \, (2\,n^2) \approx 0.796 \times (2\,n^2)$. But we should be aware of possible
factor $4/3$ overhead in the number of additive shifts per a ternary modular shift as suggested, for example, by Figure \ref{fig:modular:shift:ternary}. (Of course $4/3 \times 0.796 \approx 1.06$.)

To summarize, solutions based on emulation of binary ripple-carry adders are still win-win over the comparable true ternary ripple-carry designs in the context of the modular exponentiation; when carry
lookahead adders are used, the two options have nearly identical non-Clifford depth numbers, but there is notable width reduction advantage (factor of $\log_3(2)$) of using true ternary solution over the
emulated binary one.

\subsection{Circuits for quantum Fourier transform} \label{subsec:Fourier:transform}

In the solutions for period finding discussed so far, the quantum cost is dominated by the cost of modular exponentiation represented by an appropriate reversible classical circuit. In this context
just a fraction of the cost falls onto the quantum Fourier transform.
Nevertheless, for the sake of completeness we discuss some designs for emulating binary quantum Fourier transform on ternary computers and implementing ternary Fourier transform directly in ternary logic.

Odd radix Fourier transforms appeared in earlier quantum algorithm literature. In particular \cite{Zalka2006} outlines the benefits of ``trinary'' (ternary) Fourier for low-width Shor
factorization circuits and also briefly sketches how ternary Fourier transform can be emulated in multi-qubit framework. On a more general level, Ref. \cite{HallgrenHales} describes quantum Fourier transform over
$\mathbb{Z}_p$. In Subsection \ref{subsubsec:true:ternary:Fourier} we develop specific circuitry for a version of such a transform over $\mathbb{Z}_p$ where $p$ is some integer power of $3$.

\subsubsection{The case of emulated binary}
A familiar binary circuit for approximate Fourier transform in dimension $2^n$ with precision $\delta$ consists of roughly $\Theta(n\,\log(n/\delta))$ controlled phases and $n$ binary Hadamard gates (see
\cite{IkeAndMike2000}, Section 5).
In known fault-tolerant binary frameworks, the phases  $e^{\pi \, i/2^k}, k \in \mathbb{Z}$ occurring in the Fourier transform have to be treated just like generic phases.
Of all the possible ways to emulate a controlled phase gate we will focus on just one with minimal parametric cost.
This is the one with one clean ancilla, two Toffoli gates and one uncontrolled phase gate. (It is not clear when exactly this design has been invented, but c.f. \cite{Diagonal}, Section 2 for a more recent
discussion.)

Given the control qubit $|c\rangle$ and target qubit $|t\rangle$ the controlled phase gate $C(P(\phi)), |\phi|=1$ is emulated by applying $\mbox{Toffoli} (I\otimes I\otimes P(\phi)) \, \mbox{Toffoli}$ to the
state $|c\rangle |t\rangle |0\rangle$.
Ternary emulation of Toffoli gate is discussed in detail in Section \ref{sec:reflections}. Somewhat surprisingly, ternary emulation of uncontrolled phase gates in practice incurs larger overhead than emulation
of classical gates.
Also the binary Hadamard gate is a Clifford gate in the binary framework, but cannot be emulated by a ternary Clifford circuit. This introduces additional overhead factor of $(1+\Theta(1/\log(1/\delta)))$.
%We discuss this in more detail in subsections that follow.

\subsubsection{The case of true ternary} \label{subsubsec:true:ternary:Fourier}

We develop our own circuitry for QFT over $\mathbb{Z}_{3^n}$ based on the textbook Cooley Tukey procedure.

Quantum Fourier transform in the $n$ qutrit state space is given by the unitary matrix

\begin{equation}
\mbox{QFT}_{3^n} = [\zeta_{3^n}^{j\,k}]
\end{equation}
\noindent where $\zeta_{3^n}^{j\,k}$ is the $3^n$-th root of unity.
In particular the $\mbox{QFT}_{3}$ coincides with the ternary (Clifford) Hadamard gate.

The following recursion for $n>1$ is verified by straightforward direct computation:

\begin{equation}
\mbox{QFT}_{3^n} = \Pi_n \mbox{QFT}_{3^{n-1}} (\Lambda(D_n)) \mbox{QFT}_{3}
\end{equation}
\noindent where $\Pi_n$ is a certain $n$-qutrit permutation,

\begin{equation}
D_n = \mbox{diag}(1, \zeta_{3^n}, \ldots, \zeta_{3^n}^{3^{n-1} - 1})
\end{equation}
\noindent and where $\Lambda$ is the ternary control.

By further direct computation we observe that
\begin{equation}
D_n = \prod_{k=0}^{n-2} \mbox{diag}(1,\zeta_{3^n}^{3^k},\zeta_{3^n}^{2 \times 3^k}).
\end{equation}

The permutation gate $\Pi_n$ is not computationally important, since it amounts to $O(n)$ qutrit swaps which are all ternary Clifford.

Aside of this tweak we have decomposed $\mbox{QFT}_{3^n}$ recursively into $\Theta(n^2)$ gates of the form $\Lambda(\mbox{diag}(1,\zeta_{3^m}^{3^k},\zeta_{3^m}^{2 \times 3^k}))$ which are ternary analogs of
familiar controlled phase gates.

Similar to the binary case, it is known in general (cf. \cite{HallgrenHales}) that once we are allowed to approximate the QFT to some fidelity $1-\delta$, we can compute the approximate QFT with $\Theta(n \,
\log(n/\delta)+\log(1/\delta)^2)$ gates. This is because controlled phase gates with phases in some $O(\delta/n)$ can be dropped from the circuit without compromising the fidelity.

\subsubsection{Implementation of binary and ternary controlled phase gates in the Clifford+$R_{|2\rangle}$ basis}
In ternary framework a $P(\phi)=|0\rangle \langle 0| + \phi \, |1\rangle \langle 1|, |\phi|=1$ can be emulated exactly by the balanced two-level gate
$P'(\phi)=|0\rangle \langle 0| + \phi \, |1\rangle \langle 1| + \phi^{-1} \, |2\rangle \langle 2|$ which is a composition of the Clifford reflection $H^2$ and the non-Clifford reflection
$P''(\phi)=|0\rangle \langle 0| + \phi \, |1\rangle \langle 2| + \phi^{-1} \, |2\rangle \langle 1|$.
Also, the binary Hadamard gate $h=(|0\rangle \langle 0| + |0\rangle \langle 1|+|1\rangle \langle 0|-|1\rangle \langle 1|)/\sqrt{2}$ is a two-level Householder reflection.
As per \cite{5isNew8},\cite{BCKW}, both $P''(\varphi)$ and $h$ can be effectively approximated to precision $\delta$ by Clifford+$R_{|2\rangle}$ circuits with $R$-counts $\leq C \, \log_3(1/\delta) +
O(\log(\log(1/\delta)))$ and the constant $C$ in between $5$ and $8$.
For reference, in the Clifford+T framework the $T$-count of $\delta$-approximation of a generic phase gate is in
$3 \, \log_2(1/\delta) + O(\log(\log(1/\delta)))$.

\begin{comment}
An easy variation of Lemma \ref{lem:dual:binary:control} below shows that a controlled phase gate can be emulated at the cost of two $C_2(\mbox{INC})$ gates (a total of two $P_9$ gates) and one clean ancilla.
For reference, in Clifford+T framework ancilla-assisted emulation of controlled phase gate requires two additional modified Toffoli gates (a total of $8$ $T$ gates).
\end{comment}

Thus, emulation of the binary circuit for a binary Fourier transform incurs no surprising costs.

In pure ternary encoding we need to implement the ternary analog of controlled phase gate: gates of the form $\Lambda(\mbox{diag}(1,\phi,\phi^2)), \, |\phi|=1$. This is not difficult after some
algebraic manipulation:

\begin{prop}
Given a phase factor $\phi,|\phi|=1$  and an arbitrarily small $\delta > 0$ the gate $\Lambda(\mbox{diag}(1,\phi,\phi^2))$ can be effectively approximated to precision $\delta$ by a metaplectic circuit with at
most $40\, (\log_3(1/\delta) + O(\log(\log(1/\delta))))$ $R_{|2\rangle}$ gates.

Alternatively such a $\delta$-approximation can be effectively achieved by a metaplectic circuit with at most $24\, (\log_3(1/\delta) + O(\log(\log(1/\delta))))$ $R_{|2\rangle}$ gates and a fixed-cost widget
with at most $30$ $P_9$ gates.
\end{prop}
\begin{proof}
We note that

\begin{equation}
\begin{split}
\Lambda(\mbox{diag}(1,\phi,\phi^2)) = \phi\, \mbox{diag}(1,1,1,\phi^*,1,\phi,1,1,1)\\
 \mbox{diag}(1,1,1,1,1,1,(\phi^*)^2,1,\phi^2) (\mbox{diag}(\phi^*,1,\phi) \otimes I).
\end{split}
\end{equation}

Each of the three factors in this decomposition is a product of two two-level reflections. It is also notable that one particular reflection, the $\tau_{|0\rangle,|2\rangle}$ coming from
$\mbox{diag}(\phi^*,1,\phi) = \tau_{|0\rangle,|2\rangle} (\phi\, |0\rangle \langle 2| + |1\rangle \langle 1| + \phi^*\, |2\rangle \langle 0|)$ is in fact ternary Clifford. Therefore we are having a total of
five non-Clifford reflections in this decomposition, two of which are non-parametric classical reflections.

As per \cite{BCKW} any two-level reflection can be effectively $(\delta/5)$-approximated by  metaplectic circuit with at most $8\, (\log_3(1/\delta) + O(\log(\log(1/\delta))))$ $R_{|2\rangle}$ gates, and this
can be applied to all five non-Clifford reflections.
Alternatively, each of the two classical ones can be represented exactly as per \cite{BCRS} using five $C_2(\mbox{INC})$ or, respectively, $15$ $P_9$ gates.
\end{proof}

Thus implementation of either version of QFT circuit is never a cost surprise in the metaplectic Clifford+$R_{|2\rangle}$ basis.

Although numerologically the $R$-depth of the required approximation circuits is a good factor higher than the $T$-depth of corresponding circuits required in the Clifford+T framework, we need to keep in mind
that the $R_{|2\rangle}$ is significantly easier to execute on a natively metaplectic computer since, unlike the $T$ gate it does not require magic state distillation.

\subsubsection{Implementation of binary and ternary controlled phase gates in the Clifford+$P_9$ basis}

%Unfortunately emulating QFT on a generic ternary computer currently appears more taxing.
At the time of this writing emulation of QFT circuits on a generic ternary computer is not entirely straightforward.

First of all, we currently do not know an efficient direct circuit synthesis method for Householder reflections in the Clifford+$P_9$ basis.
If follows from \cite{Bourgain} that any ternary unitary gate can be also approximated to precision $\delta$ by an ancilla-free Clifford+$P_9$ circuit of depth in $O(\log(1/\delta))$; but  we do not have a
good effective procedure for finding ancilla-free circuits of this sort, neither do we have a clear idea of the practical constant hidden in the $O(\log(1/\delta))$.

As a bridge solution, we show in Appendix \ref{app:sec:R:2} that the requisite magic state $|\psi\rangle$ (see eq. (\ref{eq:def:psi})) for the gate $R_{|2\rangle}$ can be emulated exactly and coherently
by a set of effective repeat-until-success circuits with four ancillary qutrits and expected average $P_9$-count of $27/4$.
Thus we can approximate a required uncontrolled phase gate with an efficient Clifford+$R_{|2\rangle}$ circuit and then transcribe the latter into a corresponding ancilla-assisted probabilistic circuit over
the Clifford+$P_9$ basis. In order to have a good synchronization with the Clifford+$R_{|2\rangle}$ circuit execution it would suffice to have the magic state preparation coprocessor of width somewhat greater than
$27$. Since the controlled phase gates and hence the approximating Clifford+$R_{|2\rangle}$ circuits are performed sequentially in the context of the QFT, this coprocessor is shared across the QFT circuit and
thus the width overhead is bound by a constant.

On the balance, we conclude that ternary execution of the QFT is likely to be more expensive in terms of required non-Clifford units, than, for example, comparable Clifford+T implementation. However the
non-Clifford depth overhead factor over Clifford+T is upper bounded by an $(\alpha+\Theta(1/\log(1/\delta)))$ where $\alpha$ is a small constant.
Such overhead becomes practically valid, however, when hosting period-finding solutions that make heavy use of Fourier transform, such as for example the Beauregard circuit \cite{Beauregard} (see Appendix
\ref{app:sec:alternative:circuits} for a further brief discussion).

\subsection{Comparative cost of ternary emulation vs. true ternary arithmetic}
With the current state of the art ternary arithmetic circuits, modular exponentiation (and hence Shor's period finding) is practically less expensive with emulated binary encoding in low width (e.g. small
quantum computer); however, when $O(m^2 \log(m))$ depth is desired, pure ternary arithmetic allows for width reduction by a factor of $\log_3(2)$ compared to emulated binary circuits, while requiring
essentially the same non-Clifford depth.

\section{Implementing Reflections on Generic Ternary and Metaplectic Topological Quantum Computers} \label{sec:reflections}

State of the art implementation of the three-qubit binary Toffoli gate assumes the availability of the Clifford+T basis \cite{IkeAndMike2000}. It has been known for quite some time cf. \cite{AmyEtAl} that a
Toffoli gate can be implemented ancilla-free using a network of $\mbox{CNOT}$s and $7$ $T^{\pm 1}$ gates. It has been shown in \cite{Tcount}  that this is the minimal $T$-count for ancilla-free implementation
of the Toffoli gate.

In Section \ref{subsec:classical:with:P9} we develop emulations of classical two-level reflections (which generalize Toffoli and Toffoli-like gates) on generic ternary computer endowed with the
Clifford+$P_9$ basis as described in Section \ref{subsec:P9:gate}. We also introduce purely ternary tools necessary for implementing controlled versions of key gates for ternary arithmetic proposed in \cite{BCRS}.
%In particular we point out that any $n$-qutrit classical two-level reflection with $n>1$ on binary data can be implemented by a network of ternary Clifford gates and $6$ $P_9$ gates using $n-2$ clean
ancillas. This implies of course an emulation of the three-qubit Toffoli gate with $6$ $P_9$ gates and one clean ancilla.
%We currently do not have a proof that these emulations are optimal given the number of ancillas they use.

In Section \ref{subsec:reflection:metaplectic}
 we reevaluate the emulation cost assuming a \emph{metaplectic topological quantum computer} (MTQC) with Clifford+$R_{|2\rangle}$ basis as described in Section \ref{subsec:metaplectic:basis}.
In that setup we get two different options both for implementing non-Clifford classical two-way transpositions (including the Toffoli gate) and for circuitizing key gate for proper ternary arithmetic.

One is direct approximation using Clifford+$R_{|2\rangle}$ circuits. The other is based on the $P_9$ gate but it uses \emph{magic state preparation} in the Clifford+$R_{|2\rangle}$ basis instead of magic state
distillation. This is explained in detail in Subsection \ref{subsec:reflection:metaplectic}.
The first option might be ideal for smaller quantum computers. It allows circuits of fixed widths but creates implementation circuits for Toffoli gates with the $R$-count of approximately
$8\,\log_3(1/\delta)$ when $1-\delta$ is the desired fidelity of the Toffoli gate.
The second option supports separation of the cost of the $P_9$ gate into the ``online'' and ``offline'' components (similar to the Clifford+T framework) with the ``online'' component depth in $O(1)$ and the
``offline'' cost offloaded to a state preparation part of the computer, which has the width of roughly $9 \, \log_3(1/\delta)$ qutrits but does not need to remain always coherent.

\subsection{Implementing classical reflections in the Clifford+$P_9$ basis} \label{subsec:classical:with:P9}

The synthesis described here is a generic ternary counterpart of the exact,  constant $T$-count representation of the three-qubit Toffoli gate in the Clifford+T framework.

One distinction of the ternary framework from the binary one is that not all two-qutrit classical gates are Clifford gates.
In particular the $\tau_{|10\rangle,|11\rangle}$ reflection which is a strict emulation of the binary $\mbox{CNOT}$ is not a Clifford gate; neither is the $\tau_{|10\rangle,|01\rangle}$ which which is a
strict emulation of the binary $\mbox{SWAP}$. However, while binary $\mbox{SWAP}$ can be emulated simply as a restriction of the (Clifford) ternary swap on binary subspace, the $\mbox{CNOT}$ cannot be so
emulated.

A particularly important two-qutrit building block is the following non-Clifford gate
\[
C_1(\mbox{INC}) |j\rangle |k\rangle = |j\rangle |(k + \delta_{j,1}) \; {\rm mod} \; 3\rangle.
\]

A peculiar phenomenon in multi-qudit computation (in dimension greater than two) is that a two-qudit classical non-Clifford gate (such as $C_1(\mbox{INC})$) along with the $\mbox{INC}$ gate is universal for
the ancilla-assisted reversible classical computation, cf. \cite{Brennen}, whereas a three-qubit gate, such as Toffoli is needed for the purpose in multi-qubit case.

The following is a slight variation of a circuit from \cite{BCRS}:
\begin{equation*}\nonumber \label{eq:five:C1INC}
\begin{split}
\tau_{|02\rangle, |2,0\rangle} =
\mbox{TSWAP}\, C_1(\mbox{INC})_{2,1} \,  C_1(\mbox{INC})_{1,2} \\  \,  C_1(\mbox{INC})_{2,1}
\,  C_1(\mbox{INC})_{1,2} \,  C_1(\mbox{INC})_{2,1},
\end{split}
\end{equation*}
where $\mbox{TSWAP}$ is the ternary (Clifford) swap gate. This suggests using $5$ copies of $C_1(\mbox{INC})$ gate for implementing a two-level two-qutrit reflection. However, this is
inefficient when we only need to process binary data.

\begin{prop} \label{prop:CNOT:ancilla-free}
The following classical circuit is an exact emulation of the binary $\mbox{CNOT}$ gate on the binary data:
\begin{equation} \label{eq:two:C1INC}
\begin{split}
 \mbox{SUM}_{2,1} (\tau_{|1\rangle,|2\rangle} \otimes \tau_{|1\rangle,|2\rangle})
 \,\mbox{TSWAP} \,   C_1(\mbox{INC})_{2,1} \\ \,  C_1(\mbox{INC}^{\dagger})_{1,2}
 \, (\tau_{|1\rangle,|2\rangle} \otimes \tau_{|1\rangle,|2\rangle}) \mbox{SUM}^{\dagger}_{2,1}
\end{split}
\end{equation}
\end{prop}
\begin{proof}
By direct computation.
\end{proof}

The two non-Clifford gates in this circuit are the $C_1(\mbox{INC})$ and $C_1(\mbox{INC}^{\dagger})$. (To avoid confusion, note that the gate as per Eq. (\ref{eq:two:C1INC}) is no longer an axial reflection on ternary data.)

The $C_1(\mbox{INC})$ is Clifford-equivalent to the $C_1(Z) = diag(1,1,1,1,\omega_3,\omega_3^2,1,1,1)$ gate ($\omega_3 = e^{2 \pi\, i/3}$), and the latter gate is represented exactly by the network
shown in Figure \ref{fig:CZ:from:P9} (up to a couple of local $\tau_{|0\rangle |1\rangle}$ gates and a local $Q$ gate).

\begin{figure*}%[ht]
\begin{tikzpicture}[scale=2.000000,x=1pt,y=1pt]
\filldraw[color=white] (0.000000, -7.500000) rectangle (195.000000, 22.500000);
% Drawing wires
% Line 2: a W
\draw[color=black] (0.000000,15.000000) -- (195.000000,15.000000);
% Line 3: b W
\draw[color=black] (0.000000,0.000000) -- (195.000000,0.000000);
% Done with wires; drawing gates
% Line 6: b P Z a
\draw (12.000000,15.000000) -- (12.000000,0.000000);
\begin{scope}
\draw[fill=white] (12.000000, 0.000000) circle(6.000000pt);
\clip (12.000000, 0.000000) circle(6.000000pt);
\draw (12.000000, 0.000000) node {Z};
\end{scope}
\filldraw (12.000000, 15.000000) circle(1.500000pt);
% Line 7: =
\draw[fill=white,color=white] (30.000000, -6.000000) rectangle (45.000000, 21.000000);
\draw (37.500000, 7.500000) node {$\sim$};
% Line 8: b X
\begin{scope}
\draw[fill=white] (63.000000, -0.000000) +(-45.000000:8.485281pt and 8.485281pt) -- +(45.000000:8.485281pt and 8.485281pt) -- +(135.000000:8.485281pt and 8.485281pt) -- +(225.000000:8.485281pt and 8.485281pt)
-- cycle;
\clip (63.000000, -0.000000) +(-45.000000:8.485281pt and 8.485281pt) -- +(45.000000:8.485281pt and 8.485281pt) -- +(135.000000:8.485281pt and 8.485281pt) -- +(225.000000:8.485281pt and 8.485281pt) -- cycle;
\draw (63.000000, -0.000000) node {$P_9$};
\end{scope}
% Line 9: b P $\mbox{INC}$ a
\draw (87.000000,15.000000) -- (87.000000,0.000000);
\begin{scope}
\draw[fill=white] (87.000000, 0.000000) circle(6.000000pt);
\clip (87.000000, 0.000000) circle(6.000000pt);
\draw (87.000000, 0.000000) node {$\mbox{INC}$};
\end{scope}
\draw[fill=white] (87.000000, 15.000000) circle(1.500000pt);
% Line 10: b X
\begin{scope}
\draw[fill=white] (111.000000, -0.000000) +(-45.000000:8.485281pt and 8.485281pt) -- +(45.000000:8.485281pt and 8.485281pt) -- +(135.000000:8.485281pt and 8.485281pt) -- +(225.000000:8.485281pt and 8.485281pt)
-- cycle;
\clip (111.000000, -0.000000) +(-45.000000:8.485281pt and 8.485281pt) -- +(45.000000:8.485281pt and 8.485281pt) -- +(135.000000:8.485281pt and 8.485281pt) -- +(225.000000:8.485281pt and 8.485281pt) -- cycle;
\draw (111.000000, -0.000000) node {$P_9$};
\end{scope}
% Line 11: b P $\mbox{INC}$ a
\draw (135.000000,15.000000) -- (135.000000,0.000000);
\begin{scope}
\draw[fill=white] (135.000000, 0.000000) circle(6.000000pt);
\clip (135.000000, 0.000000) circle(6.000000pt);
\draw (135.000000, 0.000000) node {$\mbox{INC}$};
\end{scope}
\draw[fill=white] (135.000000, 15.000000) circle(1.500000pt);
% Line 12: b X
\begin{scope}
\draw[fill=white] (159.000000, -0.000000) +(-45.000000:8.485281pt and 8.485281pt) -- +(45.000000:8.485281pt and 8.485281pt) -- +(135.000000:8.485281pt and 8.485281pt) -- +(225.000000:8.485281pt and 8.485281pt)
-- cycle;
\clip (159.000000, -0.000000) +(-45.000000:8.485281pt and 8.485281pt) -- +(45.000000:8.485281pt and 8.485281pt) -- +(135.000000:8.485281pt and 8.485281pt) -- +(225.000000:8.485281pt and 8.485281pt) -- cycle;
\draw (159.000000, -0.000000) node {$P_9$};
\end{scope}
% Line 13: b P $\mbox{INC}$ a
\draw (183.000000,15.000000) -- (183.000000,0.000000);
\begin{scope}
\draw[fill=white] (183.000000, 0.000000) circle(6.000000pt);
\clip (183.000000, 0.000000) circle(6.000000pt);
\draw (183.000000, 0.000000) node {$\mbox{INC}$};
\end{scope}
\draw[fill=white] (183.000000, 15.000000) circle(1.500000pt);
% Done with gates; drawing ending labels
% Done with ending labels; drawing cut lines and comments
% Done with comments
\end{tikzpicture}
\caption{\label{fig:CZ:from:P9} Exact representation of $C_1(Z)$ in terms of $P_9$ gates. }
\end{figure*}
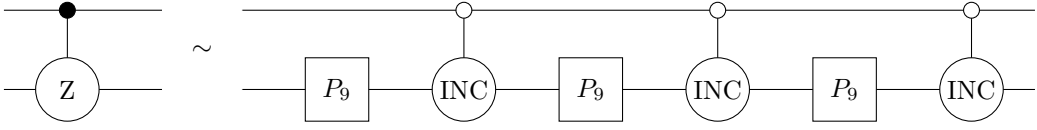

Plugging in corresponding representations of $C_1(\mbox{INC})$ and $C_1(\mbox{INC}^{\dagger})$ into the circuit (\ref{eq:two:C1INC}) we obtain an exact emulation of
$\mbox{CNOT}$ that uses $6$ instances of the $P_9^{\pm 1}$ gate.
\begin{remark} \label{remark:C1Z:depth:one}
By using an available clean ancilla, we can exactly represent the $C_1(Z)$ in $P_9$-depth one. The corresponding circuit is equivalent to one shown in Figure \ref{fig:C1Z:depth:one}. Thus the $\mbox{CNOT}$
gate can be emulated on binary data using a clean ancilla in $P_9$-depth two.
\end{remark}

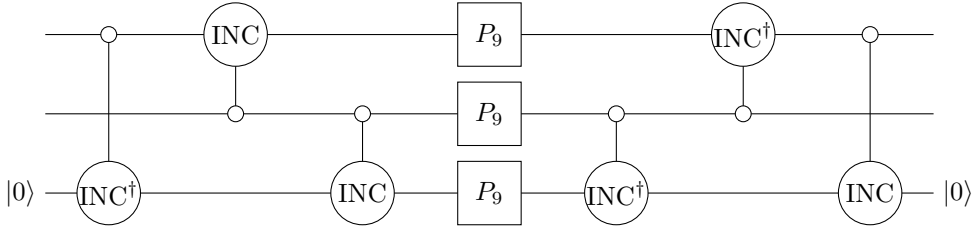
\begin{figure*}%[ht]
\begin{tikzpicture}[scale=2.000000,x=1pt,y=1pt]
\filldraw[color=white] (0.000000, -7.500000) rectangle (168.000000, 37.500000);
% Drawing wires
% Line 3: a W $$
\draw[color=black] (0.000000,30.000000) -- (168.000000,30.000000);
\draw[color=black] (0.000000,30.000000) node[left] {$$$$};
% Line 4: b W $$
\draw[color=black] (0.000000,15.000000) -- (168.000000,15.000000);
\draw[color=black] (0.000000,15.000000) node[left] {$$$$};
% Line 5: c W |0\rangle |0\rangle
\draw[color=black] (0.000000,0.000000) -- (168.000000,0.000000);
\draw[color=black] (0.000000,0.000000) node[left] {$|0\rangle$};
% Done with wires; drawing gates
% Line 7: c P $\mbox{INC}^{\dagger}$ a
\draw (12.000000,30.000000) -- (12.000000,0.000000);
\begin{scope}
\draw[fill=white] (12.000000, 0.000000) circle(6.000000pt);
\clip (12.000000, 0.000000) circle(6.000000pt);
\draw (12.000000, 0.000000) node {$\mbox{INC}^{\dagger}$};
\end{scope}
\draw[fill=white] (12.000000, 30.000000) circle(1.500000pt);
% Line 8: a P $\mbox{INC}$ b
\draw (36.000000,30.000000) -- (36.000000,15.000000);
\begin{scope}
\draw[fill=white] (36.000000, 30.000000) circle(6.000000pt);
\clip (36.000000, 30.000000) circle(6.000000pt);
\draw (36.000000, 30.000000) node {$\mbox{INC}$};
\end{scope}
\draw[fill=white] (36.000000, 15.000000) circle(1.500000pt);
% Line 9: c P $\mbox{INC}$ b
\draw (60.000000,15.000000) -- (60.000000,0.000000);
\begin{scope}
\draw[fill=white] (60.000000, 0.000000) circle(6.000000pt);
\clip (60.000000, 0.000000) circle(6.000000pt);
\draw (60.000000, 0.000000) node {$\mbox{INC}$};
\end{scope}
\draw[fill=white] (60.000000, 15.000000) circle(1.500000pt);
% Line 10: a X
\begin{scope}
\draw[fill=white] (84.000000, 30.000000) +(-45.000000:8.485281pt and 8.485281pt) -- +(45.000000:8.485281pt and 8.485281pt) -- +(135.000000:8.485281pt and 8.485281pt) -- +(225.000000:8.485281pt and 8.485281pt)
-- cycle;
\clip (84.000000, 30.000000) +(-45.000000:8.485281pt and 8.485281pt) -- +(45.000000:8.485281pt and 8.485281pt) -- +(135.000000:8.485281pt and 8.485281pt) -- +(225.000000:8.485281pt and 8.485281pt) -- cycle;
\draw (84.000000, 30.000000) node {$P_9$};
\end{scope}
% Line 11: b X
\begin{scope}
\draw[fill=white] (84.000000, 15.000000) +(-45.000000:8.485281pt and 8.485281pt) -- +(45.000000:8.485281pt and 8.485281pt) -- +(135.000000:8.485281pt and 8.485281pt) -- +(225.000000:8.485281pt and 8.485281pt)
-- cycle;
\clip (84.000000, 15.000000) +(-45.000000:8.485281pt and 8.485281pt) -- +(45.000000:8.485281pt and 8.485281pt) -- +(135.000000:8.485281pt and 8.485281pt) -- +(225.000000:8.485281pt and 8.485281pt) -- cycle;
\draw (84.000000, 15.000000) node {$P_9$};
\end{scope}
% Line 12: c X
\begin{scope}
\draw[fill=white] (84.000000, -0.000000) +(-45.000000:8.485281pt and 8.485281pt) -- +(45.000000:8.485281pt and 8.485281pt) -- +(135.000000:8.485281pt and 8.485281pt) -- +(225.000000:8.485281pt and 8.485281pt)
-- cycle;
\clip (84.000000, -0.000000) +(-45.000000:8.485281pt and 8.485281pt) -- +(45.000000:8.485281pt and 8.485281pt) -- +(135.000000:8.485281pt and 8.485281pt) -- +(225.000000:8.485281pt and 8.485281pt) -- cycle;
\draw (84.000000, -0.000000) node {$P_9$};
\end{scope}
% Line 13: c P $\mbox{INC}^{\dagger}$ b
\draw (108.000000,15.000000) -- (108.000000,0.000000);
\begin{scope}
\draw[fill=white] (108.000000, 0.000000) circle(6.000000pt);
\clip (108.000000, 0.000000) circle(6.000000pt);
\draw (108.000000, 0.000000) node {$\mbox{INC}^{\dagger}$};
\end{scope}
\draw[fill=white] (108.000000, 15.000000) circle(1.500000pt);
% Line 14: a P $\mbox{INC}^{\dagger}$ b
\draw (132.000000,30.000000) -- (132.000000,15.000000);
\begin{scope}
\draw[fill=white] (132.000000, 30.000000) circle(6.000000pt);
\clip (132.000000, 30.000000) circle(6.000000pt);
\draw (132.000000, 30.000000) node {$\mbox{INC}^{\dagger}$};
\end{scope}
\draw[fill=white] (132.000000, 15.000000) circle(1.500000pt);
% Line 15: c P $\mbox{INC}$ a
\draw (156.000000,30.000000) -- (156.000000,0.000000);
\begin{scope}
\draw[fill=white] (156.000000, 0.000000) circle(6.000000pt);
\clip (156.000000, 0.000000) circle(6.000000pt);
\draw (156.000000, 0.000000) node {$\mbox{INC}$};
\end{scope}
\draw[fill=white] (156.000000, 30.000000) circle(1.500000pt);
% Done with gates; drawing ending labels
\draw[color=black] (168.000000,0.000000) node[right] {$|0\rangle$};
% Done with ending labels; drawing cut lines and comments
% Done with comments
\end{tikzpicture}
\caption{\label{fig:C1Z:depth:one} Exact representation of $C_0(Z)$ in $P_9$-depth one. }
\end{figure*}

Thus when depth is the optimization goal, a clean ancilla can be traded for triple compression in non-Clifford depth of ternary emulation of the $\mbox{CNOT}$. (This rewrite is similar in nature to the one
employed in \cite{Jones} for the binary Margolus-Toffoli gate.)

\begin{prop} \label{prop:Toffoli:the:15}
A three-qubit Toffoli gate can be emulated, ancilla-free, by the following three-qutrit circuit:

\begin{equation} \label{eq:Toffoli:the:15}
(\mbox{SUM}^{\dagger} \otimes I) (I \otimes \tau_{|20\rangle,|21\rangle}) (\mbox{SUM} \otimes I)
\end{equation}
This circuit requires $15$ $P_9$ gates to implement.
\end{prop}
\begin{proof}
The purpose of the emulation is perform the $|110\rangle \leftrightarrow |111\rangle$ reflection in the binary data subspace.

Having applied the rightmost $\mbox{SUM} \otimes I$ we find that $(\mbox{SUM} \otimes I) |110\rangle = |120\rangle$, $(\mbox{SUM} \otimes I) |111\rangle = |121\rangle$ and we note that the latter two are the
only two transformed binary basis states to have the second trit equal to 2. Therefore the $I \otimes \tau_{|20\rangle,|21\rangle}$ operator affects only these two transformed states. By uncomputing the
$\mbox{SUM} \otimes I$ we conclude the emulation.
\end{proof}

Importantly and typically we can reduce the emulation cost by using a clean ancilla. To this end we first prove the following
\begin{lem} \label{lem:dual:binary:control}
Let $U$ be $n$-qubit unitary and let the binary-controlled $(n+1)$-qubit unitary $C(U)$ be emulated in the binary subspace of an $m$-qutrit register $m>n$. Then one level of binary control can be added to emulate
$C(C(U))$ in an $(m+2)$-qutrit register using $6$ additional $P_9$ gates; one of the new qutrits is a clean ancilla in state $|0\rangle$ and the other new qutrit emulates the binary control.

With one more ancilla the additional $P_9$ gates can be stacked to  $P_9$-depth $2$.
\end{lem}
\begin{proof}
We prove the lemma by explicitly extending the emulation circuit.
Let $c_1$ be a label of the qutrit emulating the control wire of $C(U)$. Let $c_2$ be a label of the new qutrit to emulate the new control wire. Let $a$ be the label of the new clean ancilla.

Apply the sequence of gates $C_2(INC)_{c2,a} \mbox{SUM}_{c1,c2}$ (right to left) then use the ancilla $a$ as the control in the known emulation of $C(U)$, then unentangle:
$\mbox{SUM}_{c1,c2}^{\dagger} C_2(INC)_{c2,a}^{\dagger}$.

The circuit applies correct emulation to the binary subspace of the $(m+2)$-qutrit register.
The correctness is straightforward: within the binary subspace $\mbox{SUM}_{c1,c2}$ generates $|2\rangle$ on the $c_2$ wire. The $C_2(INC)_{c2,a}$ promotes the ancilla to $|1\rangle$ if and only if
$|c_1,c_2\rangle = |11\rangle$. Therefore $U$ is triggered only by the latter basis element, which is the definition of the dual binary control.

The cost estimate follows from the fact that $C_2(INC)_{c2,a}$ and its inverse take $3$ $P_9$ gates each.
\end{proof}

\begin{corol} \label{corol:Toffoli:the:12}
Three-qubit Toffoli gate can be emulated in four qutrits (allowing one clean ancilla) with $12$ $P_9$ gates at $P_9$-depth of $4$.
\end{corol}

Indeed $\mbox{Toffoli} = CC(\mbox{NOT})$ and $C(\mbox{NOT})$ takes $6$ $P_9$ gates with no ancillas to emulate as per Proposition \ref{prop:CNOT:ancilla-free}.

\begin{corol} \label{corol:ctrl:Toffoli:18}
Four-qubit binary-controlled Toffoli gate $CCC(\mbox{NOT})$ can be emulated

1) in six qutrits (allowing two clean ancillas) with $18$ $P_9$ gates at $P_9$-depth of $6$.

2) in five qutrits (allowing one clean ancilla) with $21$ $P_9$ gates.
\end{corol}

\begin{proof}

For 1),  we emulate using Lemma \ref{lem:dual:binary:control} and Corollary \ref{corol:Toffoli:the:12}.

For 2), we emulate using Lemma \ref{lem:dual:binary:control} and Proposition \ref{prop:Toffoli:the:15}
\end{proof}

We will further use the three-qutrit ``Horner'' gate $\Lambda(\mbox{SUM})$:
\[
\Lambda(\mbox{SUM}) |i,j,k\rangle = |i,j,k+i\,j \; \mod 3 \;\rangle, \, i,j,k \in \{0,1,2\}
\]
\noindent as a tool for adding levels of control to emulated binary and true ternary gates.

Recall from \cite{BCRS}, Figure 18 and discussion, that the best-known non-Clifford cost of $\Lambda(\mbox{SUM})$ is ${4}$ $P_9$ gates at $P_9$-depth ${2}$.

We now proceed to implement the completely ternary four-qutrit gate $\Lambda \Lambda (\mbox{SUM})$ using the same constuction as above

\begin{prop} \label{prop:Lambda:Lambda:SUM}
Label primary qutrits with $1,2,3,4$ and label a clean ancillary qutrit in state $|0\rangle$ with $5$.
Then the following circuit implements the $\Lambda \Lambda (\mbox{SUM})$ gate on the primary qutrits:
\begin{equation} \label{eq:Lambda:Lambda:SUM}
\Lambda(\mbox{SUM})_{1,2,5}^{\dagger} \Lambda(\mbox{SUM})_{3,5,4} \Lambda(\mbox{SUM})_{1,2,5}
\end{equation}
This circuit requires $12$ $P_9$ gates to implement.
\end{prop}

However, as follows from discussion in Sections \ref{subsec:carry:lookahead} and \ref{subsec:ripple:carry}, controlled ternary modular exponentiation also relies on another form of the doubly-controlled
$\mbox{SUM}$ gate: the strictly controlled Horner gate $C_f(\Lambda(\mbox{SUM})), f\in \{0,1,2\}$ where the Horner gate  $\Lambda(\mbox{SUM})$ is activated only by the basis state $|f\rangle$ of the
top qutrit.

A certain implementation of the $C_f(\mbox{SUM})$ has been developed in \cite{BCRS} costing $15$ $P_9$ gates.

The following Proposition explains how to insert another level of ternary control using a cascade of Horner gates again

\begin{prop} \label{prop:C:Lambda:SUM}
Label primary qutrits with $1,2,3,4$ and label a clean ancillary qutrit in state $|0\rangle$ with $5$.
Then the following circuit implements the $C_f(\Lambda(\mbox{SUM})))$ gate on the primary qutrits:
\begin{equation} \label{eq:C:Lambda:SUM}
\Lambda(\mbox{SUM})_{2,3,5}^{\dagger} C_f(\mbox{SUM})_{1,5,4} \Lambda(\mbox{SUM})_{2,3,5}
\end{equation}
This circuit takes $23$ $P_9$ gates to implement.

With one additional ancilla the circuit can be restacked to have $P_9$-depth of $9$.
\end{prop}

Let us give a direct proof for transparency
\begin{proof}
By definition, given a four-qutrit state $|j,k,\ell,m\rangle$, we must have $C_f(\Lambda(\mbox{SUM})))|j,k,\ell,m\rangle = |j,k,\ell,m+\delta_{j,f} \, k \, \ell\rangle$.

After applying the rightmost Horner gate to the clean ancilla we have the ancilla in the $|k \, \ell\rangle$ state. The correctness of (\ref{eq:C:Lambda:SUM}) now follows from the definition of
$C_f(\mbox{SUM})_{1,5,4}$.

The best known circuitry for the components yield the cost of $15+2\times 4 = 23$ $P_9$ gates.
\end{proof}

\begin{comment}
TODO:THESE MIGHT BE STILL IMPORTANT
We start with some observations, that codify some useful commutants involving the $P_9$ gate:

\begin{observ} \label{observ:P9:commute}
1) $\mbox{INC}\,P_9 = P_9 \, (Q_0 \, \mbox{INC} $ up to global phase

2) $\tau_{|1\rangle,|2\rangle} \, P_9 = P_9^2 \,(Q_2^{\dagger} \, \tau_{|1\rangle,|2\rangle}$ up to global phase

3)  $\mbox{INC}^{\dagger} \,P_9^2  = P_9^2 \, (Q_2^{-2} \, \mbox{INC}^{\dagger}$ up to global phase
\end{observ}
These relations are established by direct computation.
\end{comment}

\subsection{Implementing classical reflections in metaplectic Clifford+$R_{|2\rangle}$ basis} \label{subsec:reflection:metaplectic}

It has been shown in \cite{BCKW} that, given a small enough $\delta>0$ any $n$-qutrit two-level Householder reflection can be approximated effectively and efficiently to precision $\delta$ by a
Clifford+$R_{|2\rangle}$ circuit containing at most $8\,\log_3(1/\delta) + O(\log(\log(1/\delta)))+ O((2+\sqrt{5})^n)$ instances of the $R_{|2\rangle}$ gate.
In particular, when $n=1$ the asymptotic term $O((2+\sqrt{5})^n)$ resolves to exactly $1$ and when $n=2$ it resolves to exactly $4$. In both cases it is safe to merge this term with the
$O(\log(\log(1/\delta)))$ term.

The single-qutrit $P_9$ gate is the composition of the ternary Clifford gate $\tau_{|0\rangle,|2\rangle}$ and the Householder reflection
$\omega_9 \, |0\rangle \langle 2| + |1\rangle \langle 1| + \omega_9^{-1} \, |2\rangle \langle 0|$.
The two-qutrit gate  $\mbox{CNOT} = \tau_{|10\rangle, |11\rangle}$ is by itself a two-level Householder reflection $R_{(|10\rangle- |11\rangle)/\sqrt{2}}$.
Similarly, $\mbox{Toffoli} = \tau_{|110\rangle, |111\rangle} = R_{(|110\rangle- |111\rangle)/\sqrt{2}}$.
Therefore, our results apply and we have efficient strict emulations of $P_9$,  $\mbox{CNOT}$ and Toffoli gates at depths that are logarithmic in $1/\delta$ and in practice are roughly
$8\,\log_3(1/\delta)$ in depth.

We note that the direct metaplectic approximation of classical reflections is significantly more efficient than the circuits expressed in $C_f(\mbox{INC})$ gates (as each of the latter have to be approximated).

Let us briefly review such direct approximation in the context of ternary arithmetic in ternary encoding. As per \cite{BCRS}, the generalized carry gate of the ternary ripple-carry additive shift contains two
classical non-Clifford reflections (\cite{BCRS}, Fig. 5) that can be represented at fidelity $1-\delta$ by a metaplectic circuit of $R$-count at most ${16}\,\log_3(1/\delta)$.

The same source implies that the carry status merge widget $\mathcal{M}$ which is key in the carry lookahead additive shift is Clifford-equivalent to a $C_f(\mbox{SUM})$ which is easily decomposed in four
classical two-level reflections and thus can be represented at fidelity $1-\delta$ by a metaplectic circuit of $R$-count at most ${32}\,\log_3(1/\delta)$.

A sufficient per-gate precision $\delta$ may be found in $O(1/(d\,\log(n)))$ where $d$ is the depth of the modular exponentiation circuit expressed in non-Clifford reflections. Therefore, injecting metaplectic
circuits in place of reflections creates an overhead factor in $\Theta(\log(d) \log(\log(n)))$. While being asymptotically moderate, such overhead could be a deterrent when factoring very large numbers.
This motivates us to explore constant-depth approximations of classical reflections as in the next section.

\subsection{Constant-depth implementation of $\mbox{CNOT}$ and $C_f(\mbox{INC})$ on ternary quantum computers.}

We demonstrate that integer arithmetic on a ternary quantum computer can be efficient both asymptotically and in practice.
We build on Section \ref{subsec:classical:with:P9} that describes exact emulation of $\mbox{CNOT}$ with $6$ instances of the $P_9$ gate.
A core result in \cite{CampbellEtAl} implies that the $P_9$ gate can be executed exactly by a deterministic state injection circuit using one ancilla, one measurement and classical feedback,
\emph{provided} availability of the ``magic'' ancillary state
\[
\mu = \omega_9^{-1} \, |0\rangle  + |1\rangle + \omega_9 \, |2\rangle.
\]
The state injection circuit is given in Figure \ref{fig:mu:state:injection}.

Assuming, hypothetically, that the magic state $\mu$ can be prepared in a separate ancillary component of the computer (then teleported), we get, a separation of the quantum complexity into ``online'' and ``offline'' components -
similar to one employed in the binary Clifford+T network.

We call these components the \emph{execution} and \emph{preparation} components. We use the term execution depth somewhat synonymously to ``logical circuit depth''.
The execution part of the $P_9$ state injection, hence $\mbox{CNOT}$,  Toffoli emulations as well as implementation of $C_f(\mbox{INC})$ are constant depth. The magic preparation can run separately in parallel
when the preparation code is granted enough width.

%The main bottleneck of this solution is synchronization between the ``online'' and ``offline'' components. The throughput of the ``offline'' component must be sufficiently high in order to yield the required
%instances of the magic state on a quantum clock schedule.

In the context of the binary Clifford+T network, assuming the required fidelity of the $T$ gate is $1-\delta, \delta>0$, there is a choice of magic state distillation solutions. For comparison we have selected
a particular one protocol described in \cite{BravyiKitaev}. At the top level, it can be described as a quantum code of depth in $O(\log(\log(1/\delta)))$ and width of approximately
$O(\log^{\log_3(15)}(1/\delta))$.
The newer protocol in \cite{BravyiHaah} achieves asymptotically smaller width in
$O(\log^{\gamma(k)}(1/\delta))$ where $k$ is an error correction hyperparameter, and $\gamma(k) \rightarrow \log_2(3)$ when $k \rightarrow \infty$. However the \cite{BravyiHaah} is a tradeoff rather than a win-win over
\cite{BravyiKitaev} in terms of practical width value.

In comparison, the magic state distillation for a generic ternary quantum computer, described in \cite{CampbellEtAl} maps onto quantum processor of depth in $O(\log(\log(1/\delta)))$ and width of
$O(\log^{3}(1/\delta))$.
Therefore the preparation of a magic state by distillation requires asymptotically larger width than the one for Clifford+T basis.
%(As a concrete example, it can be about 5 times larger when the target $\delta$ is around $10^{-21}$ and fidelity of raw magic states is $90\%$ in both cases.)

We observe that the prepartion width is asymptotically better at $O(\log(1/\delta))$ and significantly better in practice when the target ternary computer is MTQC. Since the MTQC implements the universal Clifford+$R_{|2\rangle}$
basis that does not require magic state distillation, the instances of the magic state $\mu$ can be prepared on a much smaller scale.
\begin{observ} {(see \cite{5isNew8}, Section 4)}
An instance of magic state $\mu$ can be prepared at fidelity $1-\delta$ by a  Clifford+$R_{|2\rangle}$ circuit of non-Clifford depth in $r(\delta) = 6\, \log_3(1/\delta) + O(\log(\log(1/\delta)))$.
\end{observ}

To synchronize with the $P_9$ gates in the logical circuit we need to pipeline $r(\delta)$ instances of the magic state preparation circuit, so we always have a magic state at
fidelity $1-\delta$ ready to be injected into the $P_9$ protocol.

One important consequence of the synchronization requirement is that higher parallelization of non-Clifford operations reflects proportionally in an increase in width of the preparation coprocessor. %Ancor: was 'offline'

In particular, when we employ low-width circuits for Shor's period finding, such as based on ripple-carry additive shifts, then it suffices to produce a constant number of clean magic states per time step. For
example, if the ternary $C_f(\mbox{INC})$ is taken as the base classical gate and its realization shown in figure \ref{fig:C1Z:depth:one} then we need three clean magic states per a time step.

Suppose now we employ an $n$-bit quantum carry lookahead adder in the same context. In order to preserve the logarithmic time cost advantage we should be able to perform up to $n$ base reflection gates
 (such as Toffolis)  in parallel
or at least $O(n/\log(n))$ such gates in parallel on average. Thus the preparation component must deliver at least $O(n/\log(n))$  clean magic states per time step and widens the preparation component by
that factor.

\section{Platform specific resource counts} \label{sec:specific:resource:counts}

In a more conventional circuit layout for Shor's period finding, the $\approx N^2$ modular exponentiations $|a^k \; \mod N \;\rangle , k \in [1..N^2]$, are done in superposition over $k$ and the width of
such superposition trivially depends on the integer representation radix. Thus the purely ternary encoding has the width advantage with a factor of $\log_3(2)$.

However, on a small quantum computer platform a more practical approach is to use a single control (cf. \cite{ParkerPlenio} or \cite{Beauregard}, Section 2.4), which allows to iterate through
the modular exponentiations using only one additional qubit (resp. qutrit).

With this method in mind our principal focus is on the \emph{overall cost of modular exponentiation}.

\begin{comment}
This is however not what is important in the inter-platform comparison; hence for simplicity we will be assuming that register has the fixed size of $2\,n$ bits. Under certain assumptions on measurements and
QFT compilation,
 there is however an available ``one control qubit trick'' (cf. \cite{ParkerPlenio} or \cite{Beauregard}, Section 2.4), and that register can be reduced to one qubit (and, respectively, emulated with one
 qutrit). In either case, this component of the width is fixed and can be taken out of the comparison tables, which only need to contain the width comparisons for the modular exponentiation circuits.
\end{comment}

We assume that for bitsize $n$, the $\varepsilon = 1/\log(n)$ is a sufficient end-to-end precision of the
period-finding circuit.
Then the atomic precision $\delta$ per individual gate, or rather
per individual clean magic state within the circuit depends on circuit size. The circuits under comparison differ asymptotically in depth but not in size, which is in $O(n^3)$ (disregarding the slower $O(\log(n))$ terms).
We observe that $\log(1/\delta)$
is roughly $3\,\log(n)$ for the required $\delta$. It follows that the
%offline width scaling
distillation width
for one clean magic state scales like $(3\,\log_2(n))^3$ in the ternary context. In case on magic
state preparation in the metaplectic basis one needs at most $6 \times 3\, \log_3(n)= 18 \, \log_3(n)$ $R$-gate per a clean $P_9$ magic state.
 There has been a wide array of magic state distillation protocols for the Clifford+T benchmark. For practical reasons and for simplicity we have selected the Bravyi-Kitaev protocol (\cite{BravyiKitaev}) where
 the raw magic state consumption scales like $O(\log(1/\mbox{precision})^{\log_3(15)})$; $\log_3(15) \approx 2.465$. This scaling is shown in the ``preparetion width'' cells in the resource tables below.
 An attractive alternative would be the Bravyi-Haah protocol ( \cite{BravyiHaah}).  The protocol is defined by the hyperparameter $k$ of the underlying $[n,k,d]$ error correction code and requires preparation
 width in $O((\log_2(n))^{\gamma(k)})$ where  $\gamma(k) \approx \log_2((3\,k+8)/k)$. In particular for $k=8$ the protocol distills $8$ magic states simultaneously and $\gamma(k)\approx 2$. Unfortunately this
 protocol is more sensitive to the fidelity of the raw magic states and this is one of the reasons we decided not to cost it out at this time. One needs to be mindful that the scaling exponent $\gamma(k)$ can
 in principle be made smaller than $2$ under certain circumstances.

Tables \ref{table:ripple:carry} and \ref{table:carry:lookahead} contain comparative resource estimates for the modular exponentiation circuits based, respectively, on the ripple-carry additive shift and the carry lookahead
additive shift.
%The comparison is made across all the platforms and encodings discussed earlier in the paper.
For simplicity, only \emph{asymptotically dominating} terms are represented. An actual resource bound may differ by
terms of lower order w.r.t. $\log(n)$. 
 
 In addition to resource counts for ternary processing we have provided the same for Clifford+T solutions as a backdrop. In the Clifford+T basis,  resource estimate in Table \ref{table:ripple:carry} for low-width modular exponentiation on a binary quantum computer is based on \cite{HaenerEtAl} in which an implementation was given that uses $2n+2$ logical qubits. The Toffoli depth of the circuit in \cite{HaenerEtAl} can be analyzed to be bounded by $160 n^3$. Note that the Toffoli depth is equal to $T$-depth, provided that $4$ additional ancillas are available, leading to an overall circuit width of $2n+6$. The two resource estimates in Table \ref{table:carry:lookahead} for reduced-depth modular exponentiation are based on \cite{Kutin2006} and \cite{DraperSvore}: in \cite{Kutin2006} an implementation for an arbitrary coupling architecture was given that uses $3n+6 \log_2(n) + O(1)$ qubits and has a total depth of $12 n^2 + 60n \log^2_2(n) + O(n \log(n))$. This implementation is based on a gate set that has arbitrary rotations. To break this further into Clifford$+T$ operations, we require an increase in terms of depth of $4 \log_2(1/\varepsilon) = 12 \log_2(n)$ as each rotation has to be approximated with accuracy $\varepsilon \approx 1/n^3$. Up to leading order, this leads to the estimate of the circuit depth of $144 n^2 \log_2(n)$ given in the table. In \cite{DraperSvore} a Toffoli based circuit to implement an adder in depth $4 \log_2{n}$ was given that needs $4n-\omega_1(n)$ qubits, where $\omega_1$ denotes the Hamming weight of the integer $n$. As there are $O(n)$ Toffoli in parallel in this circuit, we use the implementation of a Toffoli gate in $T$-depth $3$ from \cite{AmyEtAl} to implement a single addition in $T$-depth $12 \log_2(n)$. The modular addition can be implemented then using $3$ integer additions. To implemented Shor's algorithm, we need $2n^2$ modular additions, leading to an overall $T$-depth estimate of $72n^2 \log_2(n)$.

The rightmost column of either table lists counts proportional to either the number of raw magic states or, in the case of MTQC to the number of metaplectic magic states required per a time step
of the circuit.

\begin{comment}
UNTODO: THIS IS DEPRECATED by Krysta
The inter-platform comparisons availed by the tables are quite informative as far as the size asymptotics is concerned. The comparison of ``practical'' estimates between the platform should perhaps be takes
with grain of salt as the low level cost factors are certainly dependant on specific physical devices. One might assume that the fidelity and cost of raw resource states for the $T$ gate on one side and $P_9$
gate on the other side depend on the physics of particular qubits/qutrits being used and that the fidelity and cost of the metaplectic magic states depend on a specific engineering realization of metaplectic
anyons. Discussing the physics and engineering of the actual devices is beyond the scope and mission of this paper.
\end{comment}

\begin{table*}
    \begin{tabular}{l@{\qquad}l@{\qquad}l@{\qquad}l} \hline\hline\\[-1.5ex]
\multicolumn{1}{c}{\hspace*{-2cm} Platforms}
           &  \multicolumn{1}{l}{Circuit width} &  \multicolumn{1}{l}{Circuit depth ($P_9$/$R_{|2\rangle}/T$)} & \multicolumn{1}{l}{Preparation width}\\[0.5ex]
    \hline\\[-1.5ex]
    Emulated binary, metaplectic, via $P_9$& $n+4$ & $48 \, n^3$ & $54 \times \log_3(n)$  \\[1ex]
    Section \ref{subsec:ripple:carry}, emulated binary, via $P_9$& $n+4$ & $48 \, n^3$ & $3  (3 \, \log_2(n))^3$  \\[1ex] % 6 n^2 \times 8 n = 48 n^3
    Ternary, metaplectic, via $P_9$& $2\,m - \omega_1(m)$  & $\approx 76.35 \, n^3$ & $54\times \log_3(n)$ \\[1ex] % 304 m^3 = 76.35 n^3 TODO
    Section \ref{subsec:ripple:carry}, ternary, via $P_9$& $2\,m - \omega_1(m)$ & $\approx 76.35 \, n^3$ & $3  (3 \, \log_2(n))^3$  \\[1ex] % 16 m^2 \time 19 m = 304 m^3 = 76.35 n^3 TODO
    Emulated binary, MTQC inline& $n+4$  & $432\, n^3\, \log_3(n)$ & $3$ \\[1ex] % 6 n^2 \times 72 n \log_3(n) = 432 \times n^3 \log_3(n)
    Ternary, MTQC inline& $2\,m - \omega_1(m)$  & $\approx 506.3\, n^3 \log_3(n)$ & $3$ \\[1ex] % 6 m^2 \times 14 \times 24 \times m \log_3(m) = 2016 \timed m^3 \log_3(m) \approx 506.3 n^3 \log_2(n)
     \hline\\[-1.5ex]
    Haener et al. \cite{HaenerEtAl}, Takahashi \cite{TakaKuni}                                  & $2\, n +6$ (qubits)    &         $160\,n^3$           &
    \footnote{Here $\log_2(3) < \gamma \leq \log_3(15)$ depending on practically applicable distillation protocol. $n \times$ reflects the worst case bound on the logical width of the circuit. } $\sim n \times (6 \log_2(n))^{\gamma}$ \\[1ex]
\hline\hline
    \end{tabular}
    \caption{Size comparison for low-widths modular exponentiation circuits. $n$ is the bitsize, $m = \lceil \log_3(2) \, n \rceil$, $\omega_1(.)$ is the number of $1$s in corresponding ternary or binary
    expansion. } \label{table:ripple:carry}
  \end{table*}

% DEPRECATED $ C(k,\epsilon_0) (\log_2(n))^{\gamma(k)}$ $\gamma(k)\rightarrow \log_2(3)$.

\begin{table*} %[h!]
    \begin{tabular}{l@{\qquad}l@{\qquad}l@{\qquad}l} \hline\hline\\[-1.5ex]
\multicolumn{1}{c}{\hspace*{-2cm} Circuits}
           &  \multicolumn{1}{l}{Circuit width} & \multicolumn{1}{l}{Circuit depth ($P_9$/$R_{|2\rangle}/T$)} & \multicolumn{1}{l}{p width}\\[0.5ex]
    \hline\\[-1.5ex]
    Emulated binary, metaplectic, via $P_9$& $4\,n -\omega_1(n)$ & $120\, n^2 \, \log_2(n)$ & $54\times n\,\log_3(n)$ \\[1ex]
    Section \ref{subsec:carry:lookahead}, emulated binary, via $P_9$& $4\,n -\omega_1(n)$& $120\,n^2 \,\log_2(n)$  & $12\,n\,  (3 \, \log_2(n))^3$  \\[1ex] % 6 n^2 \times 20 \log_2(n)
    Ternary, metaplectic, via $P_9$& $4\,m -\omega_1(m)$ & $\approx 127.4 \, n^2\, \log_2(n)$ & $54\times m\,  \log_3(m)$ \\[1ex] % 320 m^2 = 127.4 n^2
    Section \ref{subsec:carry:lookahead}, ternary, via $P_9$& $4\,m -\omega_1(m)$& $\approx 127.4 \, n^2\, \,\log_2(n)$  & $12\,n\,  (3 \, \log_2(n))^3$  \\[1ex] % 16 m^2 \times \log_2(m) % 320 m^2 = 127.4 n^2

    Emulated binary, MTQC inline& $3\,n -\omega_1(n)$ & $384\,n^2 \,\log_3(2) (\log_2(n))^2$ & $3\, n$ \\[1ex] % 6 n^2 \times 64 \times \log_3(2) (\log_2(n))^2
    Ternary, MTQC inline& $3\,m -\omega_1(m)$  & $\approx 1630.5\,n^2 \, \log_3(2) (\log_2(n))^2$  & $3\,m$ \\[1ex] % 16 m^2 \times 256 \times \log_3(2) (\log_2(n))^2 % 4096 m^2 \approx 1630.5 n^2
     \hline\\[-1.5ex]
    Binary, via Clifford+T, \cite{DraperSvore} & $4\,n -\omega_1(n)$ (qubits)    &      $72\, n^2\, \log_2(n)$             &
    \footnote{Here $\log_2(3) < \gamma \leq \log_3(15)$ depending on practically applicable distillation protocol.} $3\,n\,(6 \log_2(n))^{\gamma}$ \\[1ex] %72 n^2 \log_2 = 3 \times (per Toffoli) \times 2 n^2
    %\times 3 \times 4 \log_2(n)
    Binary, via Clifford+T, \cite{Kutin2006} & $3\,n +6\,\log_2(n)$ (qubits)    &      $144\, n^2 \log_2(n)$ & $3\,n\,(6 \log_2(n))^{\gamma}$ \\[1ex]
\hline\hline
    \end{tabular}
    \caption{Sizes for reduced-depth modular exponentiation circuits. $n$ is the bitsize, $m = \lceil \log_3(2) \, n \rceil$, $\omega_1(.)$ is the number of $1$s in corresponding ternary or binary expansion.}
    \label{table:carry:lookahead}
  \end{table*}

% DEPRECATED $n\times C(k,\epsilon_0) (\log_2(n))^{\gamma(k)}$ $\gamma(k)\rightarrow \log_2(3)$

%There are two main distinctions between the resource counts assembled in the tables \ref{table:ripple:carry} and \ref{table:carry:lookahead} : the online depth and the offline width. %UNTODO:DEPRECATED

The logarithmic
%online
execution
depth for integer addition
is achieved by using
%the whole point of having the
carry lookahead additive shift circuit.
However this comes at significant width cost, as the circuit performs in parallel up to $n$ (in the worst case) or roughly $n/\log_2(n)$ (on average) reflection gates.
This requires a corresponding number of magic states or metaplectic registers simultaneously, and therefore the preparation width numbers in the last column of Table \ref{table:carry:lookahead}
are multiplied by the corresponding bit size, or, respectively, trit size. This represents the critical path bound on the magic state preparation width of the solution.

In both tables the
%offline
preparation width bound for ternary processing is dominated by the width of the $C_f(\mbox{INC})$. The Tables do not exhaust the vast array of possible depth/width
tradeoffs. We have chosen to represent $C_f(\mbox{INC})$ with non-Clifford depth one as shown in Fig. \ref{fig:C1Z:depth:one}. This circuit has the
$P_9$-width of ${3}$ and requires a clean ancillary qutrit. For the ripple-carry solution the ancillary qutrit is reused and has minimal impact. For the carry lookahead
solution up to $n$ (respectively, up to $m=\lceil \log_3(2) \, n$) ancillas must be available in parallel which inflates the online width by more than 30 percent.

The fifth and sixth rows show tradeoff based on direct approximation of Toffoli gates and (controlled) $C_f(\mbox{INC})$ gates, respectively, by topological metaplectic
circuits to precision $\Theta(1/n^3)$. The topological metaplectic $R_{|2\rangle}$ gates are executed sequentially for each individual arithmetic gate. This nearly eliminates the need for the magic state
preparation, as only $3$ topological ancillas are needed at a time in the injection circuit for the $R_{|2\rangle}$ (Figure \ref{fig:R2:RUS}). This tradeoff introduces the online depth of a
subcircuit for a Toffoli gate of roughly $24 \, \log_3(n)$. A corresponding subcircuit for a $C_f(\mbox{INC})$ then has online depth of $48 \, \log_3(m)$ (two two-level reflections). For $C_f(\Lambda
\Lambda(\mbox{INC}))$, it is $192 \, \log_3(m)$ (eight two-level reflections).

We estimate the number of required controlled integer additive shifts in a modular exponentiation circuit as $6\,n^2$ ($2\,n^2$ controlled modular additions) when binary emulation is used and as $16\,m^2$
($4\,m^2$ controlled modular additions) when ternary encoding is used. These bounds define the
%online
execution depth columns in both tables.

The most significant distinction
%exhibited by the resource counts in
in the Tables \ref{table:ripple:carry}, \ref{table:carry:lookahead} is the asymptotical advantage in the
%offline width sufficient for
magic state
preparation width with the MTQC.
% is available.

%The other informative dimension of the data presented in the tables implies
There is also a tradeoff between emulated binary encoding and true ternary encoding on a ternary quantum processor. It is seen from Table
\ref{table:ripple:carry} that with ripple-carry adders (e.g., when targeting a small quantum computer) we get a moderate practical advantage in non-Clifford circuit depth when
emulating binary encoding and a small advantage in width compared to the use of true ternary encoding. This is true even accounting for the fact that the trit size $m$ is smaller than the bit size $n$ by the
factor of $\log_3(2)$.

On the other hand when carry lookahead adders are used, the difference in the overall non-Clifford circuit depth between the two encoding scenarios is insignificant, unless inline metaplectic circuits with MTQC are compiled.
But the use of true ternary encoding yields the width advantage by a factor of roughly $\log_3(2)$.
In the fifth and sixth lines of Table \ref{table:carry:lookahead}, the use of emulated binary encoding is practically better than the use of ternary encoding.
Intuitively, this is because the metaplectic circuits are reflection-oriented and best suited for direct approximation of the (controlled) Toffoli gates that are two-level reflections, whereas
ternary arithmetic gates such as $C_f(\mbox{INC})$ or Horner have to be first decomposed into several two-level reflections.

The resource bounds shown in the tables provide a great deal of flexibility in selecting a resource balance
appropriate for a specific ternary quantum computer. On a generic ternary quantum computer where
universality is achieved by distillation of magic states for the $P_9$ gate the choice of encoding and arithmetic
circuits is likely to be dictated by the size of the actual computer. When native metaplectic
topological resources are available, magic states for the $P_9$ gate are prepared asymptotically more efficiently.
Metaplectic also offers the third choice: that of bypass the $P_9$ gate
altogether and using inline metaplectic circuits instead at the cost of a factor in $O(log(\mbox{bitsize}))$
in circuit depth expansion. In this scenario using emulated binary encoding of integers is always more efficient in
practice than the use of true ternary encoding.

\section{Conclusion}

We have investigated implementations of Shor's period finding function  (\cite{Shor}) in quantum ternary logic.
We performed comparative resource cost analysis targeting two prospective quantum ternary platforms.
The ``generic'' platform uses magic state distillation as described in \cite{CampbellEtAl}
for universality. The other, one referred to as MTQC (metaplectic topological quantum computer), is a non-Abelian anyonic platform, where universality is achieved by a relatively inexpensive protocol based on
anyonic braiding and interferomic measurement \cite{CuiHongWang},\cite{CuiWang}.

On each of these platforms we considered two different logical solutions for the modular exponential circuit of Shor's period finding function: one where the integers are encoded using the binary subspace
of the ternary state space and ternary optimizations of known binary arithmetic circuits are employed;
the other ternary encoding of integers and arithmetic circuits
stemming from \cite{BCRS} are used.

On the MTQC platform we additionally consider semi-classical metaplectic expansions of arithmetic circuits; the non-Clifford depth of such a circuit is larger than the non-Clifford depth of the
corresponding classical arithmetic circuit by a factor of $O(\log(\mbox{bitsize}))$.
Notably, circuits of this type bypass the need for magic states and the $P_9$ gate entirely.

We have derived both asymptotic and practical bounds on the quantum resources consumed by the Shor's period finding function for practically interesting combinations of platform, integer encoding and modular
exponentiation. For evaluation purposes we have derived such bounds for widths and non-Clifford depths of the logical
circuits as well as for sizes of the state preparation resources that either
distill or prepare necessary magic states with the required fidelity.

%These resource bounds are summarized and briefly discussed in the section \ref{sec:specific:resource:counts}. They are based on the known best designs for ternary arithmetic developed in \cite{BCRS}, and
%technical sections of this paper. While we are reasonably certain that these designs are optimal or near-optimal, we do not have rigorous proofs of optimality at this time. Finding such proofs is an important
%direction of future research.

We find significant asymptotic and practical advantages of the MTQC platform compared to other platforms. In particular this platform
allows factorization of an $n$-bit number using the smallest possible number of $n+7$ logical qutrits at the cost of inflating the depth of the logical circuit by a logarithmic factor. In scenarios where
increasing the depth is undesirable, the MTQC platform still exhibits significant advantage in the size of the
 magic state preparation component that is linear in the bitsize of the
target fidelity (compared to cubic or near-cubic for a generic magic state distillation).

 An interesting feature of our ternary arithmetic circuits is the fact that the denser and more compact ternary encoding
 of integers does not necessarily lead to more resource-efficient period finding solutions compared to binary encoding.
 As a rule of a thumb: if low-width circuits are desired, then binary encoding of integers combined with ternary arithmetic
 gates appears more efficient both in
 terms of width and depth than a pure ternary solution. However, even a moderate ancilla-assisted depth compression,
 such as provided by carry lookahead additive shifts, tips the balance in favor of ternary
 encoding and ternary arithmetic gates.

 In summary, having a variety of encoding and logic options, provides flexibility when choosing period finding solutions for ternary quantum computers of varying sizes.

\acknowledgements
The Authors are grateful to Jeongwan Haah for useful references. We would also like to thank Tom Draper and Sandy Kutin for providing $\langle q | pic \rangle$ \cite{DraperKutin} which we used for typesetting
most of the figures in this paper. We are thankful to an anonymous Reviewer for insightful comments that inspired us to rewrite the paper into its current, more comprehensive format.

\appendix

\section{Exact emulation of the $R_{|2\rangle}$ gate in the Clifford+$P_9$ basis.} \label{app:sec:R:2}

At this time we lack a good effective classical compilation procedure for approximating non-classical unitaries by efficient ancilla-free circuits in the Clifford+$P_9$ basis.

We show here, however, that the magic state $|\psi\rangle$ of (\ref{eq:def:psi}) that produces the $R_{|2\rangle}$ gate by state injection can be prepared by certain  probabilistic measurement-assisted
circuits over the Clifford+$P_9$ basis.
Therefore the compilation into the Clifford+$P_9$ basis can be reduced to a compilation into the Clifford+$R_{|2\rangle}$ basis, while incurring a certain state preparation cost. This solution, however inelegant, is
sufficient, for example, in the context of Shor's integer factorization.

We have seen in Section \ref{subsec:classical:with:P9} that the classical $C_1(\mbox{INC})$ and, hence, the classical $C_2(\mbox{INC})$ gates can be represented exactly and ancilla-free using three $P_9$
gates. We use the availability of these gates to prove the key lemma below.

Recall that $\omega_3=e^{2 \pi \, i/3}$ is a Clifford phase.

\begin{lemma} \label{app:lemma:resource:state}
Each of the ternary resource states

\[
(|0\rangle+\omega_3 |1\rangle)/\sqrt{2}, \quad \mbox{and} \quad
(|0\rangle+\omega_3^2 |1\rangle)/\sqrt{2}
\]

\noindent can be represented exactly by a repeat-until-success (RUS) circuit over Clifford+$P_9$ with one ancillary qutrit and expected average number of trials equal to $3/2$.

\end{lemma}
\begin{proof}
Let us give a proof for the second resource state. (The proof is symmetrical for the first one.)

We initialize a two-qutrit register in the state $|20\rangle$ and compute
\[
C_2(\mbox{INC}) (H\otimes I)|20\rangle = (|00\rangle + \omega_3^2 |10\rangle + \omega_3 |21\rangle)/\sqrt{3}.
\]
If we measure $0$ on the second qutrit, then the first qutrit is in the desired state. Overwise, we discard the register and start over. Since the probability of measuring $0$ is $2/3$, the Lemma follows.
\end{proof}

\begin{corol}
A copy of the two-qutrit resource state
\begin{equation} \label{app:eq:eta}
|\eta\rangle = (|0\rangle+\omega_3 |1\rangle)\otimes (|0\rangle+\omega_3^2 |1\rangle)/2
\end{equation}
\noindent can be represented exactly by a repeat-until-success circuit over Clifford+$P_9$ with two ancillary qutrits and expected average number of trials equal to $9/4$.
\end{corol}

To effectively build a circuit for the Corollary, we stack together the two RUS circuits described in Lemma \ref{app:lemma:resource:state}.

\begin{lemma}
There exists a measurement-assisted circuit that, given a copy of resource state $|\eta\rangle$ as in (\ref{app:eq:eta}), produces a copy of the resource state
\begin{equation} \label{app:eq:psi}
|\psi\rangle = (|0\rangle- |1\rangle+|2\rangle)/\sqrt{3}
\end{equation}
\noindent with probability $1$.
\end{lemma}
\begin{proof}
Measure the first qutrit in the state $(H^{\dagger}\otimes I) \mbox{SUM} |\eta\rangle$.

Here is the list of reduced second qutrit states given the measurement outcome $m$:
\[
m=0 \mapsto (|0\rangle- |1\rangle+|2\rangle)/\sqrt{3},
\]
\[
m=1 \mapsto (|0\rangle- \omega_3 |1\rangle+ \omega_3^2 |2\rangle)/\sqrt{3},
\]
\[
m=2 \mapsto (|0\rangle- |1\rangle+\omega_3 |2\rangle)/\sqrt{3}.
\]
While the first state on this list is the desired $|\psi\rangle$, each of the other two states can be turned into $|\psi\rangle$ by classically-controlled Clifford correction.

\end{proof}

As shown in \cite{CuiWang},Lemma 5, the resource state $|\psi\rangle$ as in (\ref{app:eq:psi}) can be injected into a coherent repeat-until-success circuit of expected average depth $3$ to execute the
$R_{|2\rangle}$  gate on a coherent state. See our Figure \ref{fig:R2:RUS} in Section \ref{subsec:metaplectic:basis}.

Recall that the $C_2(\mbox{INC})$ gate appearing in the lemma \ref{app:lemma:resource:state} construction has the non-Clifford cost of three $P_9$ gates.
Thus, to summarize the procedure: we can effectively and exactly prepare the magic state $|\psi\rangle$ using four-qutrit register at the expected average $P_9$-count of $27/4$.

To have a good synchronization of the magic state preparation with the $R_{|2\rangle}$  gate injection is would suffice to have a magic state preparation coprocessor of width greater than $27$ (to compensate
for the variances in repeat-until-success circuits).

\section{Circuit fidelity requirements for Shor's period finding function} \label{app:sec:fidelity}

To recap the discussion in Section \ref{subsec:shor:top:level}, the quantum period finding function consists of preparing a unitary state $|u\rangle$ proportional to the superposition

\begin{equation}
\sum_{k = 0}^{N^2} |k\rangle |a^k \; {\rm mod} \; N\rangle
\end{equation}

followed by quantum Fourier transform, followed by measurement, followed by classical postprocessing.

As we know, the measurement result $j$ can be useful for recovering a period $r$ or it can be useless. It has been shown in \cite{Shor} that the probability $p_{\mbox{useful}}$ of getting a useful measurement
is in
$\Omega(1/(\log(\log(N))))$.

Speaking in more general terms, let $\mathcal{H}$ be the Hilbert space where the $\mbox{QFT} |u\rangle$ is to be found after the quantum Fourier transform step, let $\mathcal{G} \subset \mathcal{H}$ be the
subspace spanned by all possible state reductions after all possible useful measurements and let $\mathcal{G}^{\bot}$ be its orthogonal complement in $\mathcal{H}$.

Let $\mbox{QFT} |u\rangle = |u_1\rangle + |u_2\rangle , |u_1\rangle \in \mathcal{G}, |u_2\rangle \in \mathcal{G}^{\bot}$ be the orthogonal decomposition of $\mbox{QFT} |u\rangle$. Then $p_{\mbox{useful}}
=||u_1\rangle|^2$.

Let now $|v\rangle$ be an imperfect unitary copy of $\mbox{QFT} |u\rangle$ at Hilbert distance $\varepsilon$. What is the probability of obtaining \emph{some} useful measurement on measuring $|v\rangle$?

By definition, it is the probability of $|v\rangle$ being projected to $\mathcal{G}$ upon measurement.

\begin{prop}
In the above context the probability of $|v\rangle$ being projected to $\mathcal{G}$ upon measurement is greater than
\[
p_{\mbox{useful}} - 2 \, \sqrt{p_{\mbox{useful}}} \, \varepsilon
\]
\end{prop}

\begin{proof}
Let $|v\rangle=|v_1\rangle+|v_2\rangle, |v_1\rangle \in \mathcal{G}, |v_2\rangle \in \mathcal{G}^{\bot}$ be the orthogonal decomposition of the state $|v\rangle$.

Clearly $||u_1\rangle-|v_1\rangle| < \varepsilon$ and, by triangle inequality, $||v_1\rangle| \geq ||u_1\rangle |- ||u_1\rangle-|v_1\rangle| > ||u_1\rangle| - \varepsilon$.

Hence $||v_1\rangle|^2 > (||u_1\rangle| - \varepsilon)^2 > ||u_1\rangle|^2 - 2 \, ||u_1\rangle| \, \varepsilon =  p_{\mbox{useful}} - 2 \, \sqrt{p_{\mbox{useful}}} \, \varepsilon$
as claimed.
\end{proof}

\begin{corol}
In the above context, if $\varepsilon < \gamma \, \sqrt{p_{\mbox{useful}}}$ where $0 < \gamma < 1/2$, then the probability of obtaining \emph{some} useful measurement on measuring $|v\rangle$ is greater than
$(1-2\, \gamma) \,  p_{\mbox{useful}}$.
\end{corol}

In particular if $\varepsilon < \sqrt{p_{\mbox{useful}}}/4$ , we are at least half as likely to obtain a useful measurement from the proxy state $|v\rangle$ as from the ideal state $\mbox{QFT} |u\rangle$.

In summary, there is a useful precision threshold $\varepsilon$ in $O(1/(\sqrt{\log(\log(N))}))$ that allows to use an imprecisely prepared state at precision $\varepsilon$ in place of the ideal state in the
measurement and classical post-processing part of Shor's period finding function.

This translates into per-gate tolerance in the preparation circuit in a usual way. If $d$ is the unitary depth of the state preparation circuit, then it suffices to represent each of the consecutive unitary
gates to fidelity $1-\varepsilon/d$ or better. For completeness, we make this argument explicit in the following proposition. Let $||U||$ denote the spectral norm of a unitary operator $U$.

\begin{prop}
Assume that an ideal quantum computation $U=\prod_{k=1}^d U_k$ is specified using $d$ perfect unitary gates $U_k$ and we actually implementing it using $d$ imperfect unitary gates $V_k$ where for all
$k = 1,\ldots,d$ it holds that $||U_k-V_k|| \leq \delta$. Then for the actually implemented unitary transformation $V = \prod_{k=1}^d V_k$ it holds that $||U-V|| \leq d \, \delta$.
\end{prop}

\begin{proof} (See also  \cite{BernsteinVazir97}, \cite{Vazirani98}).
We perform induction on $d$. When $d=1$ there is nothing to prove. Assume the inequality has been proven for a product of length $d-1$.

We have $||U-V|| = ||\prod_{k=1}^d U_k - (\prod_{k=1}^{d-1} U_k)\, V_d + (\prod_{k=1}^{d-1} U_k)\, V_d - \prod_{k=1}^d V_k||$
$ \leq ||\prod_{k=1}^d U_k - (\prod_{k=1}^{d-1} U_k)\, V_d || + ||(\prod_{k=1}^{d-1} U_k)\, V_d - \prod_{k=1}^d V_k|| $
$ = ||\prod_{k=1}^{d-1} U_k|| \, ||U_d - V_d|| + ||(\prod_{k=1}^{d-1} U_k)- \prod_{k=1}^{d-1} V_k||\,||V_d|| $
$ \leq \delta + (d-1) \delta = d\, \delta $,

\noindent where in the second step we used the triangle inequality, in the third step the multiplicativity of the norm, i.e. $||U\,V||=||U||\, ||V||$ for all unitaries U,V, and that $||U||=1$ for all unitary
$U$. In the last step we used the inductive hypothesis.

\end{proof}

\section{Alternative circuits for modular exponentiation} \label{app:sec:alternative:circuits}

\subsection{Width optimizing circuits}
One way for the modular exponentiation to use qubits (resp., qutrits) sparingly is to perform computation in phase space.

First consistent solution of this kind has been proposed in \cite{Beauregard}. A peculiar feature of the proposed solution is that the modular additive shift block for $|b\rangle \mapsto |(a+b) \; {\rm mod} \;
N\rangle$ has four interleaved quantum Fourier transforms (two direct and two inverse, see Figure 5 in \cite{Beauregard}), the sole purpose of which is establishing and then forgetting the sign of $a+b-N$.
It is unlikely that any of these transforms can be made redundant without significant redesign of the circuit.
As we have pointed out in Section \ref{subsec:Fourier:transform}, ternary quantum computers are comparatively inefficient in practice when emulating non-classical
modular exponentiation circuits such as Beauregard \cite{Beauregard}.

Fortunately the Beauregard circuit
%has been superseded by Takahashi
can be supplanted by Haener et. al
circuit \cite{HaenerEtAl}.
%\cite{TakaKuni} which is classical. That circuit sidesteps the need for clean ancillas by a clever use of the %``idle''
%qubits.
Instead of emulating that circuit directly,
we point out that our ternary modular exponentiation
circuit based on ripple-carry adder (see Section \ref{subsec:ripple:carry}), maintains smaller width in qutrits by much simpler means - systematic use of ternary
basis state $|2\rangle$ which, for our purposes is always ``idle''.

\subsection{Depth-optimizing circuits}

Reversible classical circuits that stand apart from relatively simple layouts we have analyzed, are hinted at in a hidden-gem paragraph in Section 5 of \cite{CleveWatrous}.

Let us revisit equation (\ref{eq:mod:exp:product}) in Section \ref{subsec:mod:exponetiation}

\begin{equation}
a^k \; {\rm mod} \; N = \prod_{j=0}^{2\,n-1} (a^{2^j} \; {\rm mod} \; N)^{k_j}  \; {\rm mod} \; N
\end{equation}

It is pointed out in \cite{CleveWatrous} that, instead of accumulating the partial modular products sequentially, one can accumulate pairwise modular products at nodes of a binary tree, whose original leaves
are classically pre-computed values
$a^{2^j} \; {\rm mod} \; N$. This prepares the entire product in depth $O(\log(n))$ (instead of $2\,n$) multiplications and size $O(n^2)$.

Furthermore, each of the pairwise multiplications can be set up as a binary tree of modular additions and so performed in depth $O(\log(n))$ and size $O(n)$.

Thus the entire modular exponentiation is done in depth $O(\log(n)^2)$ and size $O(n^3)$.

This proposal still uses modular addition as the core building block, and thus we can plug in an emulation of the the modular addition circuit built out of carry-lookahead adders and comparators as in the
Section \ref{subsec:mod:exponetiation}. This fits well into the polylogarithmic depth promise of \cite{CleveWatrous}.

It should be pretty straightforward to rewrite this design in ternary logic and use the ternary lookahead additive shift circuits as described in subsection \ref{subsec:ternary:ternary:ternary} to circuitize
it. For the lack of time and space we have to forego a more detailed analysis of this here.

%Here we are still in the domain of the purely classical reversible circuits, where the cost of emulation of the Toffoli gate is the dominating cost factor (see our thesis Thesis \ref{thesis:first}).

\begin{comment}
%\bibliography{BGKSRbiblio}

\end{comment}

\end{document}